\documentclass[a4paper,12pt,article]{article}
\usepackage[utf8]{inputenc}
\usepackage{bm}
\usepackage {amsthm,amsmath,amssymb,amsfonts,amssymb}
\usepackage{amscd}
\usepackage{listings} 
\usepackage{booktabs}
\usepackage{array}
\usepackage{graphicx}
\usepackage{float}
\usepackage{colonequals}
\usepackage{enumerate}
\usepackage{multicol}
\usepackage{multirow}
\usepackage{mathrsfs}
\usepackage{lmodern}
\usepackage[usenames,dvipsnames]{pstricks}
\usepackage{epsfig}
\usepackage{pst-grad} 
\usepackage{pst-plot} 
\usepackage[hmargin=3cm,vmargin=2.5cm]{geometry}
\usepackage{rotating}
\usepackage{amscd}
\usepackage{fancyhdr}
\usepackage{caption}
\usepackage{subcaption}
\usepackage{pdfpages}
\usepackage{lscape}

\usepackage{amscd}
\usepackage{latexsym}
\usepackage{delarray}

\usepackage{setspace} 
\usepackage{natbib}
\usepackage[draft]{fixme}
\usepackage{mathtools}
\usepackage{hyphenat}
\usepackage[section]{placeins}

\usepackage{graphicx} 
\usepackage{enumerate} 
\usepackage{pst-grad} 
\usepackage{pst-plot} 
\usepackage{setspace}
\usepackage{fancyhdr} 
\usepackage{caption}
\usepackage[]{subcaption}
\usepackage{pdfpages}
\usepackage{lscape}

\usepackage{delarray}

\DeclareMathOperator{\cor}{cor}

\DeclareMathOperator{\Var}{Var}         
\DeclareMathOperator{\E}{E}    

\usepackage{url}



\usepackage{dcolumn}
\newcolumntype{e}{D{.}{.}{9}}

\numberwithin{equation}{section}

\theoremstyle{plain}

\newtheorem{theorem}{Theorem}[subsection]
\newtheorem{proposition}{Proposition}
\newtheorem{corollary}[theorem]{Corollary}
\newtheorem{definition}{Definition}

\numberwithin{equation}{section}

\DeclareMathOperator{\var}{Var}

\providecommand{\keywords}[1]
{
  \small	
  \textbf{\textit{Keywords---}} #1
}

\graphicspath{ {images/} }

\title{Heritability curves: A local measure of heritability}

\begin{document}

\author{Geir Drage Berentsen\thanks{Department of Business and Management Science, NHH Norwegian School of Economics. Helleveien 30, 5045 Bergen, Norway}, 
 Francesca Azzolini\thanks{Department of Mathematics, University of Bergen, P.O. Box 7803, 5020 Bergen, Norway}, 
 Hans J. Skaug\footnotemark[2], \\
 Rolv T.~Lie\thanks{Department of Global Public Health and Primary Care, University of Bergen, P.O. Box 7804, 5020 Bergen, Norway}
 \thanks{Centre for Fertility and Health, Norwegian Institute of Public Health, Oslo, Norway},
 H\aa kon K. Gjessing\footnotemark[3]
 \footnotemark[4]
 }

\maketitle

\begin{abstract}
This paper introduces a new measure of heritability which relaxes the classical assumption that the degree of heritability of a continuous trait can be summarized by a single number. This measure can be used in situations where the trait dependence structure between family members is non-linear, in which case traditional mixed effects models and covariance (correlation) based methods are inadequate. Our idea is to combine the notion of a correlation curve with traditional correlation-based measures of heritability, such as Falconer’s formula. For estimation purposes, we use a multivariate Gaussian mixture, which is able to capture non-linear dependence and respects certain distributional constraints. We derive an analytical expression for the associated correlation curve, and investigate its limiting behaviour when the trait value becomes either large or small. The result is a measure of heritability that varies with the trait value. When applied to birth weight data on Norwegian mother--father--child trios, the conclusion is that low and high birth weight are less heritable traits than medium birth weight. On the other hand, we find no similar heterogeneity in the heritability of Body Mass Index (BMI) when studying monozygotic and dizygotic twins.
\end{abstract}


\section{Introduction}



Biometrical modeling of family trait correlations has a very long
tradition, going back at least to Ronald Fisher~\citep{fisher_xvcorrelation_1919} and Sewall Wright~\citep{wright_relative_1920,wright21}, and being developed into an extensive modeling framework over
 the years~\citep{bulmer_mathematical_1985,Neale2002}, with openly available software tools, such as OpenMx~\citep{neale2016openmx}. For a continuous trait $Y$, such as weight or height, the
basic idea is that trait variability -- or more precisely, the variance
of the measured trait, $\var(Y)$ -- can be decomposed into genetic
and environmental components, each explaining a portion of the
observed trait variance. Thus, the concept of \emph{heritability}
can, loosely, be defined as the proportion of trait variance explained
by genetic components, with environmental influences assumed to explain
the rest~\citep{Hopper2002}. As an example, the most common twin model, known as the
ACE model, decomposes the trait $Y$ into additive genetic effects
(A), common (shared) environment (C), and residual (random) environment
(E). In terms of variances, we commonly define quantities $a^{2}$,
$c^{2}$, and $e^{2}$ as the \emph{proportions} of trait variances
explained by the components A, C, and E, respectively. Thus, assuming
that no other effects are present, we have $a^{2}+c^{2}+e^{2}=1$.

To separate genetic variance from environmental variance, family data
are needed. Genetic correlations between family members decrease in
more distant relationships, thus providing contrasts from which the
genetic components can be estimated. For instance, in the classical
ACE twin design, the additive genetic correlation in monozygotic twin
pairs is assumed to be 1, whereas the corresponding correlation, or
degree of shared genetic influence, is assumed to be 1/2 in dizygotic
twin pairs. In addition, it is frequently assumed that the amount
of shared environment is the same in dizygotic twins as is monozygotic
twins. The quantities $a$ and $c$ above can also be seen as the
degree to which the underlying genes $A$ and shared environmental
$C$ are being ``expressed'' in the phenotype of each individual.
Thus, the monozygotic twin pair phenotype correlation will be $\rho^{(MZ)}=a^{2}+c^{2}$,
and $\rho^{(DZ)}=\frac{1}{2}a^{2}+c^{2}$ for the dizygotic twin pairs.
As a consequence, the difference $\frac{1}{2}a^{2}$ between monozygotic
and dizygotic twin pair correlations is ascribed to genes alone, providing
an estimate of the heritability $a^{2}$.

The ACE model is very specific in its assumption of additive genetic
effects, as well as independent, additive contributions from the environment.
In the biometrical modeling literature, a wide range of variants and
extentions have been developed. Using family structures of increasing
complexity, numerous different effects can be identified, such as
additive genetic effects, dominant genetic effects, X-chromosome effects,
effects of maternal genes on the fetus during pregnancy, effects of
mitochondrial genes, gene-gene interactions, gene-environment interactions,
etc.~\citep{Neale2002,HopperVisscher2002,Gjessing08}. Extending the family structures
used for modeling is in general challenging since genetic correlations
between more distant relatives quickly drop to nearly undetectable
levels, and assumptions about how environmental factors are shared
within larger families become harder to verify~\citep{Gjessing08}.
Still, with a steady increase in registry-based population studies
with large sample sizes and available data on environmental covariates,
such modeling has become feasible.

Common to practically all models in the field is that the degree of
heritability is assumed constant across the full range of the phenotype.
For instance, the estimated proportion $a^{2}$ of variance explained
by additive genes is assumed to be the same whether the phenotype
$Y$ is small, close to its mean, or large. It seems clear, however,
that for instance rare but dramatic environmental influences on the
phenotype may occasionally cause the phenotype to deviate strongly
from its mean value, much more than would be expected under ``normal''
circumstances. Below, we illustrate our models of heritability using
a child's birth weight (BW) as phenotype. While the birth weight distribution
is close to a normal distribution, it has a heavier tail to the
left (Figure~\ref{fig:scatmat}); this may indicate a higher proportion of low birth weight children
than what would be expected from many minor genetic and environmental
components adding up during pregnancy. 

\begin{figure}[!ht] 
        \centering
\includegraphics{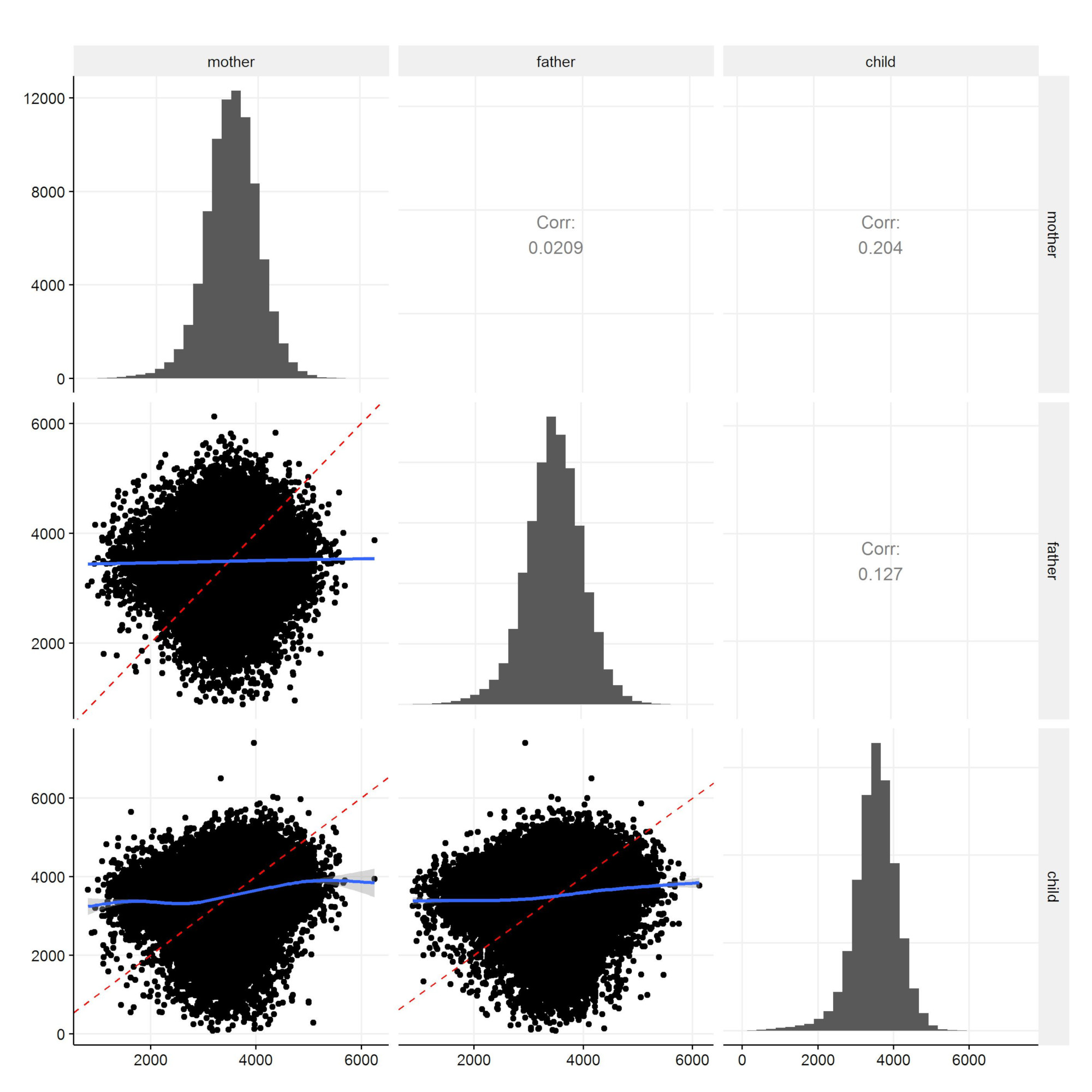}
       \caption{Birth weights (gram) for $81,144$ mother–father–child trios from the
	   Norwegian Birth Registry. Diagonal: histograms of marginal birth weights. Lower triangle: 
	   pairwise scatter plots with estimated nonparametric regression line (blue) and identity line (dashed red),
	   where $y=x$.
	   Upper triangle: pairwise empirical correlation.}
        \label{fig:scatmat}
\end{figure}

This simple observation may
suggest that 
the degree of heritability of birth weight can differ
in the different ranges of weight; perhaps the lowest BW values are
caused by ``rouge'' environmental factors that act more strongly
than genetic effects in the tail, or maybe they are caused by rare,
recessive genes that only occasionally excert a strong negative influence
on BW.

These observations motivate us to look for differences in heritability
across the range of the trait value $Y$.  The existing methods for investigating such differences are almost exclusively based on regression methods. In their seminal work~\citep{defries1985multiple}, DeFries and Fulker evaluate the degree of regression to the mean for co-twins of probands from strata in the tails of a continuous trait distribution. The idea is that if the trait is heritable, then we should observe DZ co-twins with a higher degree of regression to the mean compared to the MZ co-twins. This approach is known as DeFries-Fulker (DF) extremes analysis for twins. Later, a formal test was developed to examine whether the heritability of the trait for probands in the selected strata was equal or different to the unselected population~\citep{defries1988multiple}. This methodology was extended by Cherny et al.~\citep{cherny1992differential} by considering interaction effects between the heritability of the trait and the realized value of the trait for the proband. This approach can be used to detect linear and quadratic changes in heritability as the trait value changes. These methods all have the drawback of only providing a rough description of how the heritability varies with the trait value. The DF approach requires the researcher to select a cut-off point (a low or high trait value) for choosing the strata; the result can thus be misleading if the heritability changes smoothly as the trait value vary. Conversely, if there exists a point in the trait distribution where the heritability jumps and then stabilize again, the Cherny approach will only model this change by a linear or quadratic curve.

These drawbacks were addressed in ~\citep{logan2012heritability} using quantile regression; by using the extended DF extremes analysis ~\citep{labuda1986multiple} as the quantile regression equation, the authors obtain a heritability measure for each quantile of the trait distribution. Consequently, their method results in a heritability measure for each value of the trait $Y = y$, corresponding to a specific quantile of the distribution. 

However, in the present paper we introduce an approach based on localizing traditional genetic models. Informally, this means making sense of estimating, for instance, the additive genetic effect as
a function of the phenotype; i.e.~to define meaningfully $a^{2}(y)$
as the proportion of phenotype variance explained by additive genetic
effects, conditional on $Y=y$. Such a definition may seem
self-contradictory since one conditions on the variable whose variance
is being decomposed. Nevertheless, it is fully possible to make sense
of this concept, and we show in this paper how to develop \emph{heritability
curves}, such as $a^{2}(y)$. This definition thus provides a ``local''
measure of heritability, depending on the phenotype value.

As for the ACE twin model, all standard biometrical models rely on
the phenotype correlations between family individuals to estimate
the variance components that determine heritability. Our starting
point for developing a local measure of heritability is thus a local
measure of dependence between family members; more specifically, we
need a local measure of correlation. There are several local measures
proposed in the literature, such as the local Gaussian correlation
~\citep{Tjostheim2013}, the dependence function ~\citep{Holland1987},
and the correlation curve ~\citep{bjerve93}. We base our approach on
the correlation curve ~\citep{bjerve93} $\rho(y)$, which can be defined
as a measure of locally explained variance, and thus fits the framework
of heritability as a proportion of explained variance. The correlation
curve is similar to the traditional Pearson's correlation in that
it takes values between minus one and one, and the square $\rho^{2}(y)$
is a measure of locally explained variance. In a bivariate Gaussian
distribution, the correlation curve is constant (independent of $y$),
and equal to the standard Pearson correlation. In contrast to the
Pearson correlation the local correlation of a bivariate
relationship depends on direction; for a bivariate random variable
$(Y_{1},Y_{2})$, the locally explained variance of $Y_{2}$ conditional
on $Y_{1}=y$ may differ from the locally explained variance of $Y_{1}$
conditional on $Y_{2}=y$.

With phenotype measurements on, for instance, a mother ($Y_{1}$)
and her child ($Y_{2}$), it may seem reasonable, for instance, to
study the distribution of a child phenotype conditionally
on the maternal phenotype. However, most biometrical models are formulated
in terms of genetic and environmental factors \emph{shared} by the
two family members, thus assuming a form of exchangeability between
the two. This is particularly clear in twin pairs, where conditioning
one twin on the other twin is unnatural. In the model of ~\citep{logan_heritability_2012} this assignment was done randomly, while ~\citep{cherny_differential_1992} explored both a random assignment and a double-entry approach. However, the population value of the
correlation curve can be derived from the joint distribution of two
variables. If the joint distribution is exchangeable, so that $(Y_{1},Y_{2})$
has the same bivariate distribution as $(Y_{2},Y_{1})$, the correlation
curve is invariant to which variable we condition on, i.e.~whether
we measure the locally explained variance of $Y_{1}$ conditional
on $Y_{2}$ or vice versa. This means that the role of the mother
and child in the above interpretation can be interchanged.

The correlation curve may be estimated parametrically or non-parametrically
from observed values of a bivariate distribution $(Y_{1},Y_{2})$
by conditioning on either $Y_{1}=y$ or $Y_{2}=y$. However, our approach
is instead to first model the bivariate distribution as a Gaussian
mixture distribution, where the mixture distribution is restricted
in such a way as to be exchangeable. From the mixture distribution,
the correlation curve can be derived explicitly. We estimate the distribution
by maximum likelihood, and by allowing a sufficient number of components,
a mixture distribution is very flexible and fits a wide range of distributional shapes.
Having obtained the parameters of the mixture distribution, the correlation
curve can be derived from its explicit expression by plugging in the
estimated parameters. 

The paper is structured as follows. In Section \ref{section:Devel},
we define a standard mixed-effect model for continuous traits, and
structure it for two specific family models: twin pairs and mother--father--child
trios. Following a standard twin approach~\citep{falconer60}, and 
models for family trios~\citep{Magnus01,lunde_genetic_2007},
we derive expressions for the heritability estimates in both family
structures. In Subsections \ref{section:corr} and \ref{section: Her},
we explain the concept of correlation curves, and extend the traditional
definition of heritability to the heritability curve, which depends
on the trait value $y$. In Section \ref{section:Gauss}, we introduce
and analyze a Gaussian mixture~\citep{mclachlan2000fmm} for bivariate
phenotype distributions, parameterized to be exchangeable.
We then study the limiting behaviour of the correlation curve for large and small phenotype values under
this model in Subsection \ref{sec:asym}. Lastly, in Subsection \ref{sec:est},
we discuss the estimation of the correlation curve for the twin-pairs
and the mother--father--child trios models. Section \ref{section:Appl}
provides two applications of this approach. Namely, the first application
is the analysis of BMI values for twin pairs collected in the dataset
``twinData", found in the R-package "OpenMx"
\citep{neale2016openmx}; the second one is
the analysis of birth weight data of mother--father--child
trios from the Medical Birth Registry of Norway. For both family structures
we compute AIC and BIC values to select the best-fitting mixture models,
and explore the resulting distributions and heritability curves. 
Proofs are provided in an appendix.


\section{Development of Heritability curves}

\label{section:Devel} 



\subsection{Traditional models for twins and family trios}

 We first provide a basic description of how traditional biometrical
models can be set up in some generality, and in particular for twins
and family trios. While there are numerous ways of building, parametrizing,
and interpreting such models, our approach is fairly standard, and
in a form that supports our development of heritability curves. Let
$Y_{ij}$ be the trait value of individual $j$ in a family $i$,
and consider the mixed-effect model (see e.g. ~\citep{mcculloch01})
\begin{equation} \label{eq:mixmod1}
Y_{ij}=\mu+\beta^{t}x_{ij}+A_{ij}+C_{ij}+D_{ij}+E_{ij},
\end{equation}
where $A_{ij}$, $C_{ij}$, $D_{ij}$ and $E_{ij}$ represent additive
genetic, common environmental, dominant genetic, and residual environmental
random effects, respectively (see e.g. ~\citep{falconer60}). We assume
the four components $A_{ij}$, $C_{ij}$, $D_{ij}$ and $E_{ij}$
to be mutually independent, with mean $0$ and variances $\sigma_{A}^{2}$,
$\sigma_{C}^{2}$, $\sigma_{D}^{2}$ and $\sigma_{E}^{2}$. The inclusion
of the term $\beta^{t}x_{ij}$ (fixed effects) allows the average
phenotype level to depend on covariates. Note that this model assumes
no gene-environment interaction. In traditional biometrical modelling
(see e.g. ~\citep{Gjessing08}) the random effects are assumed to be
normally distributed with expectation $0$, i.e. $A_{ij}\sim N(0,\sigma_{A}^{2})$,
$C_{ij}\sim N(0,\sigma_{C}^{2})$, $D_{ij}\sim N(0,\sigma_{D}^{2})$
and $E_{ij}\sim N(0,\sigma_{E}^{2})$. The assumption of normality
is seen as natural based on the central limit theorem if $Y$ is the
result of numerous small, independent genetic and environmental effects
that add up to produce the trait value. Under the above assumptions the
total variance of the trait is given by 
\begin{equation}
\sigma^{2}=Var(Y_{ij})=\sigma_{A}^{2}+\sigma_{C}^{2}+\sigma_{D}^{2}+\sigma_{E}^{2}. \label{eq:varcomp}
\end{equation} 
We define $a^{2}=\sigma_{A}^{2}/\sigma^{2}$, $c^{2}=\sigma_{C}^{2}/\sigma^{2}$,
$d^{2}=\sigma_{D}^{2}/\sigma^{2}$, and $e^{2}=\sigma_{E}^{2}/\sigma^{2}$
as the proportions of the total variance that derive from each of
the four genetic and environmental components. Note that

\[
a^{2}+c^{2}+d^{2}+e^{2}=1,
\]
i.e.~the contributions from all components sum to one. Thus, in a
model including $A$, $C$, and $E$, excluding dominant effects, one
may quantify the genes-versus-environment contribution to trait variability
as $a^{2}$. This proportion is often referred to as \textit{heritability}
and can be interpreted as how strongly the genetic effect $A_{ij}$
contributes to the trait value. The heritability based on the additive
genetic component is often referred to as \emph{narrow sense heritability}.
Some models may also include dominant genetic effects, and in such
cases one may refer to $a^{2}+d^{2}$ as the \emph{broad sense heritability}~\citep{khoury_fundamentals_1993}.

From independent observations of $Y_{ij}$ alone, it is not
possible to identify the individual variance components $\sigma_{A}^{2}$,
$\sigma_{C}^{2}$, $\sigma_{D}^{2}$, and $\sigma_{E}^{2}$ in \eqref{eq:varcomp},
only the total variance $\sigma^{2}$. In order to make the individual
variances identifiable, one has to consider data on family members,
for which the $Y$'s are correlated due to shared genetic material
and environment. We focus on two basic family structures --- mother--father--child
trios and twin pairs --- in the following. As is well known, these
family structures are quite restricted in the number of effects they
allow to be estimated, and assumptions have to be made about what
genetic and environmental effects to include in each model. In the
following, we will present the specific models that will serve as
illustrations when developing heritability curves.

\subsubsection{Twins}\label{subsubsec:twins}

Perhaps the best known biometrical model is the ACE model for twins,
complemented by the alternative ADE model. While the expressions for
twin correlations in these models are very well known, we state them
here as a starting point for the heritability curves.

Let $Y_{ij}$ be the trait value of twin $j$ ($j=1,2$) in twin-pair $i$. Let
$\rho^{(MZ)}$ and $\rho^{(DZ)}$ be the phenotype correlations $\cor(Y_{i1},Y_{i2})$
for MZ and DZ twins, respectively. Both ACE and ADE models include
the additive genetic component $A$. For MZ-twins $\cor(A_{i1},A_{i2})=1$,
while for DZ-twins $\cor(A_{i1},A_{i2})=1/2$. In the standard ACE
model, the correlation for the common environmental effect is assumed
to be $\cor(C_{i1},C_{i2})=1$ in all twin pairs; thus, one makes
the common assumption of DZ twins sharing their environment to the
same degree as the MZ twins. In the alternative ADE one assumes $\cor(D_{i1},D_{i2})=1$
for MZ twins and $\cor(D_{i1},D_{i2})=1/4$ for DZ twins. In both
models, residual environmental effects are assumed to be independent.

Since the basic twin models utilize only the $\rho^{(MZ)}$ and $\rho^{(DZ)}$
phenotype correlations, they allow estimating two parameters. In addition,
$e{{}^2}$ can be estimated from $e^{2}=1-a^{2}-c^{2}-d^{2}$. The
ACE model assumes $d^{2}=0$, and thus the parameters $a^{2}$, $c^{2}$,
and $e^{2}$ can be identified; the ADE model assumes $c^{2}=0$,
and thus the parameters $a^{2}$, $d^{2}$, and $e^{2}$ can be identified. 

For the ACE model, it follows from the above that
\begin{eqnarray*}
\rho^{(MZ)} & = & a^{2}+c^{2},\\
\rho^{(DZ)} & = & \frac{1}{2}a^{2}+c^{2}.
\end{eqnarray*}
For the ADE model, the equations are
\begin{align*}
\rho^{(MZ)} & =a^{2}+d^{2},\\
\rho^{(DZ)} & =\frac{1}{2}a^{2}+\frac{1}{4}d^{2}.
\end{align*}
The simplest approach to estimating $a^{2}$, $c^{2}$, and $d^{2}$
is by moment estimators, i.e.~to solve this set of equations, using
empirical values for $\rho^{(MZ)}$ and $\rho^{(DZ)}$, and use $e^{2}=1-a^{2}-c^{2}-d^{2}$
to estimate $e^{2}$. The resulting solutions for the ACE model are
the celebrated formulas of Falconer:~\citep{falconer60}
\begin{align} \label{eq:Falconer} 
 & a^{2}=2(\rho^{(MZ)}-\rho^{(DZ)}), \nonumber \\
 & c^{2}=2\rho^{(DZ)}-\rho^{(MZ)}, \\
 & e^{2}=1-\rho^{(MZ)}.\nonumber 
\end{align} 
For the ADE model, the corresponding set of solutions are
\begin{align} \label{eq:ADE_moment}
a^{2} & =4\rho^{(DZ)}-\rho^{(MZ)},\nonumber \\
d^{2} & =2(\rho^{(MZ)}-2\rho^{(DZ)}),   \\
e^{2} & =1-\rho^{(MZ)}.\nonumber 
\end{align} 
Without further assumptions, an informal choice between the ACE and
ADE models is often made based on whether empirically $\rho^{(MZ)}<2\rho^{(DZ)}$ or
not. If this is the case, the ACE model is a natural choice; otherwise,
the ADE model can be used.

\subsubsection{Mother-father-child trios}

\label{sec:trios} Let $Y_{ij}$ be the observed trait value of individual
$j$ in nuclear family trio $i$. We let $j=1,2,3$ correspond to
the mother, father, and child, respectively. A phenotype correlation
between mother and father may signify, for instance, assortative mating,
inbreeding, or social homogamy among the parents. 
However, the correlation is typically low, and we will here assume it is zero~\citep{Magnus01}. There are thus only two correlations that provide information: the mother-child and
father-child correlations. There are numerous ways of parametrizing correlations in 
nuclear families~\citep{Magnus01,pawitan2004,lunde_genetic_2007,Gjessing08,rabe-hesketh_biometrical_2008}, but being restricted to two correlations 
means that these cannot be separated.
In our setting, we assume, for additive autosomal genes,
that $\cor(A_{i1},A_{i3})=\cor(A_{i2},A_{i3})=1/2$, and that $\cor(A_{i1}A_{i2})=0$
for the parents. Also, we assume that mother and child
share an environmental component, but no such sharing between father and child,
leading to $\cor(C_{i1},C_{i3})=1$ and $\cor(C_{i2},C_{i3})=0$.
Thus, 
\[
\Sigma_{A}=\sigma_{A}^{2}\begin{bmatrix}1 & 0 & 1/2\\
0 & 1 & 1/2\\
1/2 & 1/2 & 1
\end{bmatrix},\quad\Sigma_{C}=\sigma_{C}^{2}\begin{bmatrix}1 & 0 & 1\\
0 & 1 & 0\\
1 & 0 & 1
\end{bmatrix},\textrm{ and }\Sigma_{E}=\sigma_{E}^{2}\begin{bmatrix}1 & 0 & 0\\
0 & 1 & 0\\
0 & 0 & 1
\end{bmatrix}
\]
are the covariance matrices for the vectors $(A_{i1},A_{i2},A_{i3})$,
$(C_{i1},C_{i2},C_{i3})$, and $(E_{i1},E_{i2},E_{i3})$, respectively.

A graphical representation of the above model is displayed in a path diagram in Figure \ref{fig:pathdiagram}.

 Under the above assumptions the vectors $(Y_{i1},Y_{i2},Y_{i3})$
are \emph{i.i.d. }multivariate normal with mean

 \begin{equation} \label{eq:multmean}
 (\mu+\beta^{t}x_{i1},\mu+\beta^{t}x_{i2},\mu+\beta^{t}x_{i3}) 
 \end{equation} 
and covariance matrix 
\begin{equation}
\Sigma=\Sigma_{A}+\Sigma_{C}+\Sigma_{E}=(\sigma_{A}^{2}+\sigma_{C}^{2}+\sigma_{E}^{2})\begin{bmatrix}1 & 0 & \frac{1}{2}a^{2}+c^{2}\\
0 & 1 & \frac{1}{2}a^{2}\\
\frac{1}{2}a^{2}+c^{2} & \frac{1}{2}a^{2} & 1
\end{bmatrix},\label{eq:covtot}
\end{equation}
where $a^{2}$, $c^{2}$, and $e^{2}$ are defined as above. Again,
the unknown values can simply be estimated by the methods of moments
by matching the correlation matrix~(\ref{eq:covtot}) to its empirical
counterpart, and solve for $a^{2}$, $c^{2}$ and $e^{2}$ under the
constraint $a^{2}+c^{2}+e^{2}=1$. The solution is given by the following equations
\begin{eqnarray}
a^{2} & = & 2\rho^{(FC)}\nonumber \\
c^{2} & = & \rho^{(MC)}-\rho^{(FC)}  \label{eq:ACE_MFC}  \\
e^{2} & = & 1-\rho^{(MC)}-\rho^{(FC)},\nonumber
\end{eqnarray}
where $\rho^{(MC)}$ and $\rho^{(FC)}$ are the mother-child and father-child
correlations, respectively.

We will, in the following, use these solutions, and those for the ADE
twin model, to obtain local versions of $a^{2}$, $c^{2}$, $d^{2}$,
and $e^{2}$. Note that in both cases, the underlying assumption
is that the covariance (correlation) matrix completely characterizes
the dependence structure between traits in a family and can be decomposed
as in e.g. \eqref{eq:covtot}.

\begin{figure}[H] 
        \centering
\includegraphics{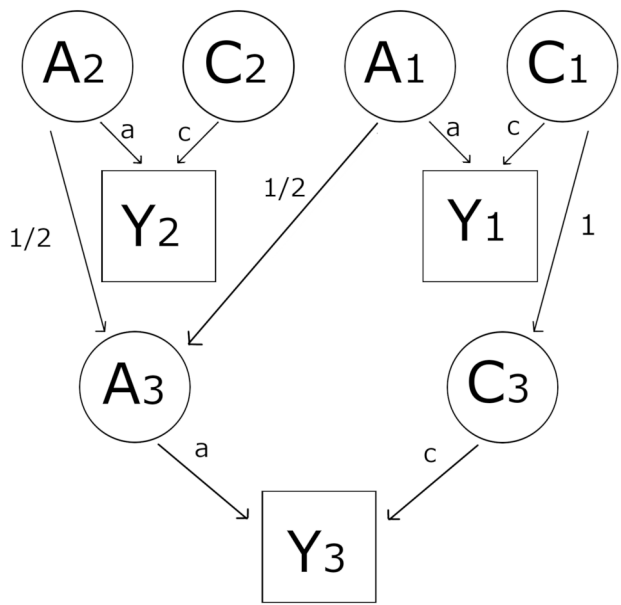}
       \caption{ Path diagram representing the birth weight of mother $Y_1$, father $Y_2$, and child $Y_3$ (represented as squares). The traits are determined by the unobserved genotype values ($A$) and 
environmental values ($C$) (shown as circles), as well as the independent residual environmental values ($E$) (not shown).}
        \label{fig:pathdiagram}
\end{figure}

\subsection{Correlation curves for non-linear bivariate relationships}

\label{section:corr} 

We now explain the concept of local correlation curves, following
the approach of Bjerve and Doksum~\citep{bjerve93}. To illustrate the principle
of localization, we use simulated data from a hypothetical phenotype,
as seen in Figure \ref{fig:stratfig}.

 We consider two strata (A and
B) consisting of all mother-child pairs for which the mother's trait
$Y_{1}=y_{1}$ falls within two intervals (interval A and B) on the
x-axis. The corresponding correlation curve is shown in Figure
\ref{fig:corfig}; as a function of $y_{1}$ (horizontal axis) it is smaller in 
stratum A than in stratum B. This indicates that the mother-child
association is stronger in stratum B compared to stratum A. In a non-parametric
regression setting, this would mean that the child's trait can be
predicted by the mother's trait with higher precision in stratum B
than in stratum A. For both strata, an increase in the mother's
trait is associated with an increase in the child's trait since the
correlation curve is positive. Since the correlation curve is continuous,
the location argument $y_{1}$ can be seen as the center of infinitesimal
intervals from which strata such as A and B can be constructed, while
the value of the correlation curve is a measure of dependence for
the corresponding strata. A constant correlation curve indicates that
the dependence properties are constant across these strata, while a
varying correlation curve indicates strata that differ in their dependence
properties.

\begin{figure}[H]
\caption{Illustration of the concept of a correlation curve and the role
of exchangeability using simulated data. Strata A and B include all mother-child pairs for which the mother's trait value falls in the intervals $[1,2]$ and $[9, 10]$, respectively. Strata $A^*$ and $B^*$ include all mother-child pairs for which the child's trait value falls in the same intervals.}

\captionsetup[subfigure]{font = large}
\begin{subfigure}[T]{0.45\linewidth}
       \centering
\includegraphics[width = \linewidth]{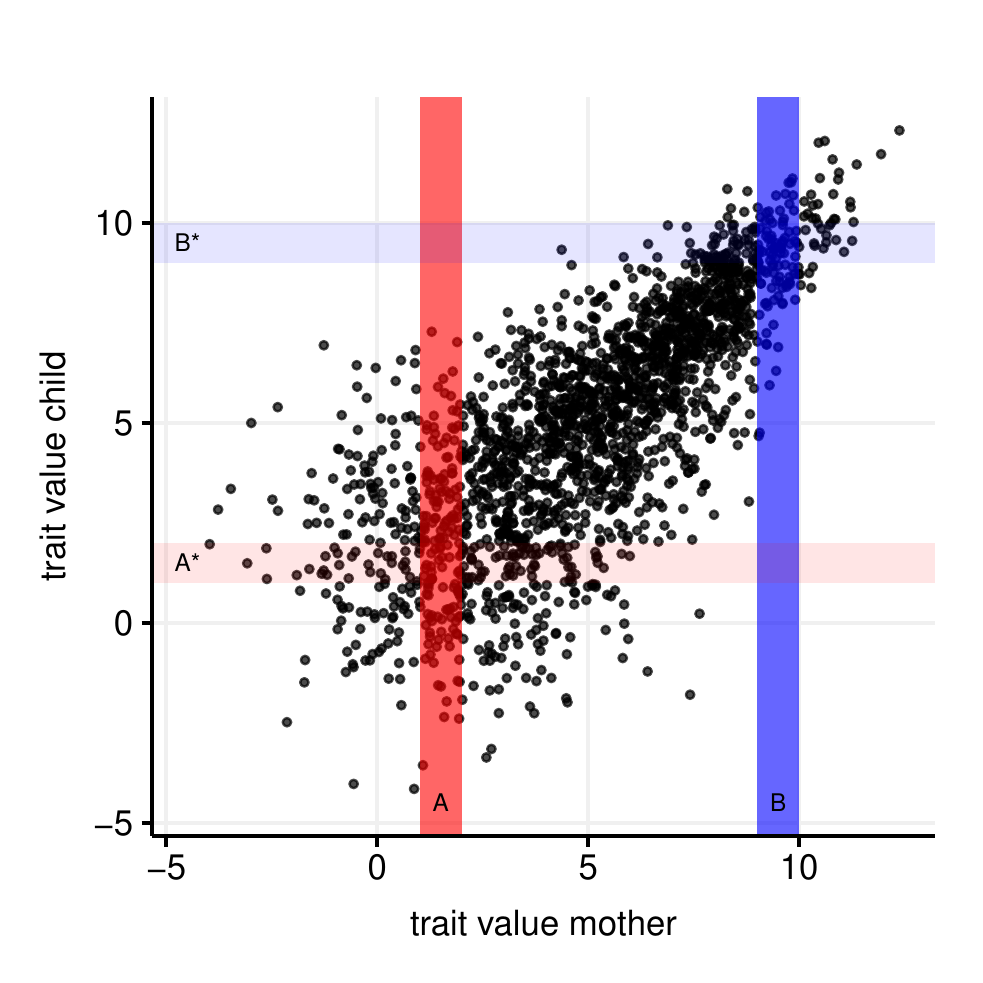}
      \subcaption{Simulated data from an exchangeable Gaussian mixture, 
	  and the definition of strata.}
       \label{fig:stratfig}
\end{subfigure}
\quad
\begin{subfigure}[T]{0.45\linewidth}
       \centering
\includegraphics[width = \linewidth]{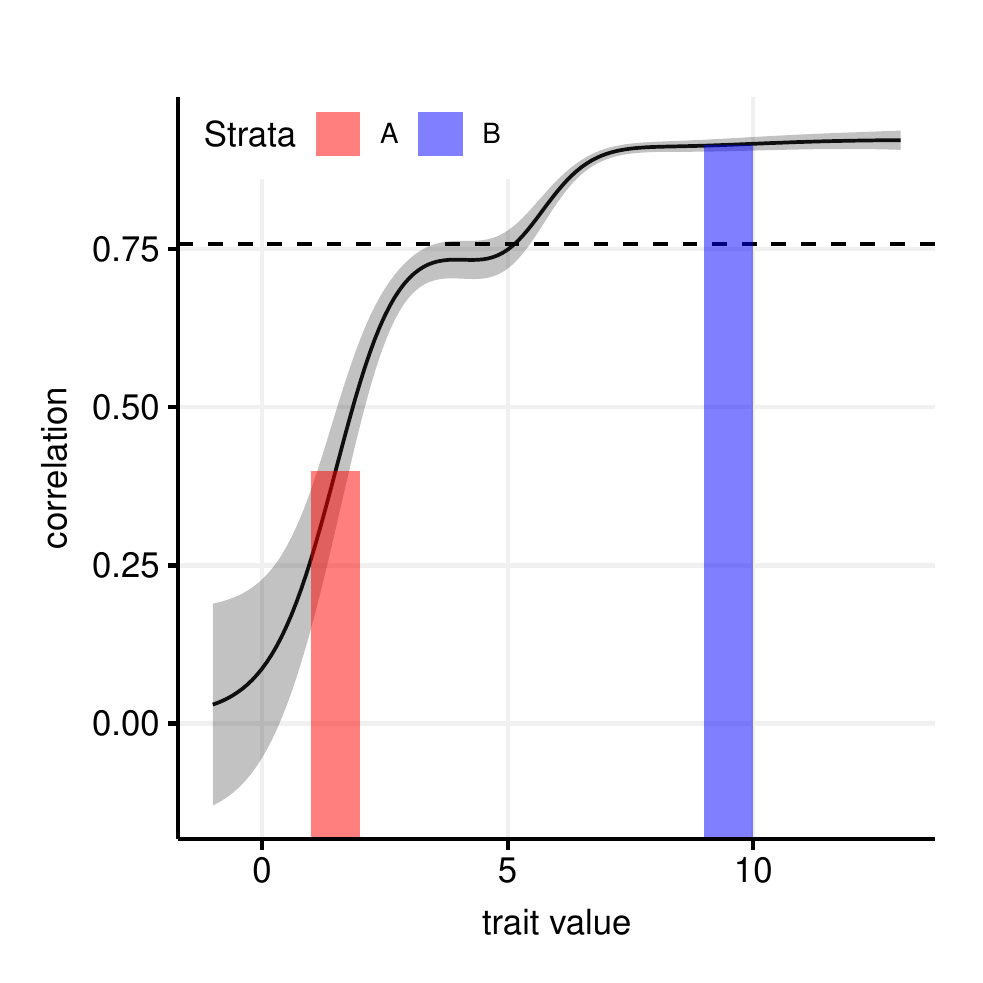}
      \subcaption{Estimated correlation curve for the data displayed in panel (a), and 95\% pointwise confidence intervals are shown in grey. The height of the bars displays the average value of the correlation curve within strata $A$ and $B$.}
       \label{fig:corfig}
\end{subfigure}
\end{figure}

If the joint distribution is exchangeable, so that $(Y_{1},Y_{2})$
has the same bivariate distribution as $(Y_{2},Y_{1})$, the correlation
curve is invariant to which variable we condition on, i.e.~whether
we measure the locally explained variance of $Y_{1}$ conditional
on $Y_{2}$ or vice versa. This means that the role of the mother
and child in the above interpretation can be interchanged, and the dependence
structure in strata $A^{*}$ and $B^{*}$ in Figure \ref{fig:stratfig})
is similar to the dependence structure in strata $A$ and $B$; the
correlation curve $\rho(y)$ as a function of $y$ thus represents
a measure of the mother-child trait dependence when either the mother
or the child has trait value equal to $y$. In the next section, we
show more precisely how $\rho(y)$ is defined in terms of locally
explained variance.

\subsubsection{Standard correlation curves for bivariate relationships}

Let $(Y_{1},Y_{2})$ be random variables from a bivariate continuous
distribution, and define $\tau_{1}^{2}=\Var(Y_{1})$, $\tau_{2}^{2}=\Var(Y_{2})$,
and $\rho=\cor(Y_1,Y_2)$. Further, define $\mu(y)=\E(Y_{1}|Y_{2}=y)$ and $\sigma^{2}(y)=\Var(Y_{1}|Y_{2}=y)$ as functions of $y$. Assuming that $\mu(y)$ is differentiable, define $\beta(y)=\mu'(y)$,
i.e.~the slope of the (typically non-linear) regression curve $\mu(y)$
when $Y_{1}$ is regressed on $Y_{2}$. Recall that in a standard linear regression context, $\mu(y)$ is a linear function of $y$, where the slope $\beta_{1|2} := \beta(y)$ and the conditional variance $\sigma_{1|2}^2 := \sigma^2(y)$ are both constant.

By the law of total variance,
\[
\Var(Y_{1})=\Var(\E(Y_{1}|Y_{2}))+\E(\Var(Y_{1}|Y_{2})),
\]
and it thus seems natural to define in general
\[
\text{Proportion of }\Var(Y_1)\text{ explained by }Y_2=\frac{\Var(\E(Y_{1}|Y_{2}))}{\Var(\E(Y_{1}|Y_{2}))+\E(\Var(Y_{1}|Y_{2}))}.
\] 
In the case of linear regression, $\Var(\E(Y_{1}|Y_{2}))= \tau_2^2\beta_{1|2}^2$
and $\E(\Var(Y_{1}|Y_{2}))= \sigma_{1|2}^2$, and the proportion of explained variance can thus be written
\begin{equation}
\frac{(\tau_2\beta_{1|2})^2}{(\tau_2\beta_{1|2})^2+\sigma_{1|2}^2}=\left(\frac{\tau_2\beta_{1|2}}{\tau_1}\right)^2=\rho^2,\label{eq:prop.expl.lin}
\end{equation}
which is the usual formula for explained variance in a linear regression.

We want to define a ``local'' variant of $\rho^{2}$, describing
the proportion of explained variance when $Y_{2}=y$, thus to define
$\rho^{2}(y)$ as a function of $y$. To this end, (\ref{eq:prop.expl.lin}) is a natural starting point, and the extension to a non-linear setting would thus be to allow both $\beta(y)$ and $\sigma^2(y)$ to depend on $y$. This leads to the definition
\begin{equation}
\rho(y)=\frac{\tau_{2}\beta(y)}{\left[\left(\tau_{2}\beta(y)\right)^{2}+\sigma^{2}(y)\right]^{1/2}},\label{correlationcurve}
\end{equation}
where we recall that $\tau_{2}^{2}=\Var(Y_{2})$, $\beta(y)=\frac{d}{dy}\,\E(Y_{1}|Y_{2}=y)$,
and $\sigma^{2}(y)=\Var(Y_{1}\mid Y_{2}=y)$.

Indeed, this is the formula developed by Bjerve et al.~\citep{bjerve93}
and Doksum et al.~\citep{doksum_correlation_1994}. As pointed
out by Bjerve et al., the correlation curve should not be confused
with the conditional correlation obtained by applying the usual correlation
formula to the conditional distribution of $(Y_{1},Y_{2})$ given
$Y_{2}=y$, which would always be zero. It should also be noted that while $\tau_2$ is kept fixed in~(\ref{correlationcurve}), the denominator $\left(\tau_{2}\beta(y)\right)^{2}+\sigma^{2}(y)$ is no longer necessarily equal to $\tau_1^2=\Var(Y_1)$ from the original distribution. In fact, for a \emph{fixed} $y=y_0$, it corresponds to $\Var(Z_1)$ from a hypothetical bivariate distribution $(Z_1,Z_2)$ where $\Var(Z_2) = \tau_2^2$ and $\Var(Z_1)$ is determined from having a linear regression of $Z_1$ on $Z_2$ with constant slope $\beta(y_0)$ and constant conditional variance $\Var(Z_1|Z_2) = \sigma^2(y_0)$.

\subsubsection{Correlation curves for symmetric bivariate relationships}

In our setting, we are interested in relationships between pairs of
family members, for example, a pair of twins or a child and a parent.
We denote the pair's respective trait values by $Y_{1}$ and $Y_{2}$.
At first glance, it may seem natural to ask about the explained variation
of a child trait $Y_{1}$, conditional on its parental value $Y_{2}$.
However, this is less natural for twins, who are from the same generation.
Indeed, most biometrical models assume that the positive correlation
between the trait values is generated by shared genes and shared environment;
the sharing is symmetrical between family members, and the generational
aspect is only used to compute the degree of relatedness. That is,
in pairs of family members, the two members should be exchangeable,
so that $(Y_{1},Y_{2})$ and $(Y_{2},Y_{1})$ have the same bivariate
distribution. Clearly, this means that when applying (\ref{correlationcurve})
in a heritability setting, it would be reasonable to expect that $Y_{1}$
conditional on $Y_{2}$ should provide the same answers as $Y_{2}$
conditional on $Y_{1}$. While exchangeability is obviously not the
case for general bivariate distributions, we achieve pairwise exchangeability
by a corresponding restriction of our parametric models for the bivariate
distributions, as described later. When including covariates, the
assumption of pairwise exchangeability should apply to the residuals,
i.e.~the mean-adjusted traits $Y_{1}-\beta^{t}x_{1}$ and $Y_{2}-\beta^{t}x_{2}$.

Note that it would suffice to assume that, for all $y$,
\begin{align} \label{def:exchangability}
\tau_1^2=\Var(Y_{1})&=\Var(Y_{2})=\tau_2^2=:\sigma, \nonumber\\
\E(Y_{1}\mid Y_{2}=y)&=\E(Y_{2}\mid Y_{1}=y) =: \mu(y),\\
\Var(Y_{2}\mid Y_{1}=y)&=\Var(Y_{1}\mid Y_{2}=y) =: \sigma^2(y), \nonumber
\end{align}
since this would imply that (\ref{correlationcurve}) would be invariant
to the direction of conditioning. However, the models presented in
this paper all imply full pairwise exchangeability. We do \emph{not},
however, ask for full exchangeability of the multivariate outcome
distribution; for instance, a mother-father-child trio would clearly
not have the same trivariate distribution as a child-father-mother
trio. Nevertheless, the pairwise exchangeability implies that all
family members have the same marginal distributions.
The appropriateness of the exchangeability assumptions will be addressed
in the Discussion.


\subsection{Heritability curves}

\label{section: Her} 

Assuming $\rho(y)$ to be well defined for the joint distribution
of the two family members, we are interested in the degree to which the value
of $\rho(y)$ can be attributed to heritability on one side, and to
environment on the other. In particular, we are interested in knowing
how these contributions vary with $y$. 

\begin{definition}[Heritability curve for the twin ADE model] \label{def:heritability_curve_twins} Assume the exchangeability property~(\ref{def:exchangability}) holds for both MZ and DZ bivariate distributions. Adopting the moment equations \eqref{eq:ADE_moment}, we define
the heritability curve by
\begin{equation}
a^{2}(y)=4\rho^{(DZ)}(y)-\rho^{(MZ)}(y),\label{eq:a(y)_twins}
\end{equation}
where $\rho^{(MZ)}(y)$ and $\rho^{(DZ)}(y)$ are the correlation
curves of MZ and DZ twins calculated according to (\ref{correlationcurve}).
Similarly, (\ref{eq:ADE_moment}) allows local versions
of the dominance effect
\begin{equation}
d^{2}(y) =2\left[(\rho^{(MZ)}(y)-2\rho^{(DZ)}(y)\right]
\end{equation} 
and residual environment
\begin{equation}
e^{2}(y) =1-\rho^{(MZ)}(y)
\end{equation} 
to be defined.
\end{definition} 

Note that with Equation~\eqref{eq:a(y)_twins}, a trait value can
in principle display a non-linear association within both MZ and DZ
twins, but have constant local heritability $a^{2}(y)$ due to a canceling
effect in $4\rho^{(DZ)}(y)-\rho^{(MZ)}(y)$.

We similarly define the heritability curve for family trios by adopting
the genetic model described in Section \ref{sec:trios} locally.

\begin{definition}[Heritability curve for an ACE model of mother-father-child
trios] \label{def:heritability_curve_trios} Assuming the exchangeability 
property~\eqref{def:exchangability},
let $\rho^{(MC)}(y)$ and $\rho^{(FC)}(y)$ be correlation curves~(\ref{correlationcurve})
for mother-child and father-child relationships, respectively. The
heritability curves $a^{2}(y)$, $c^{2}(y)$, and $e^{2}(y)$ are then
given by 
\begin{align}
a^{2}(y) & =2\rho^{(FC)}(y)\label{eq:a(y)_triplets}\\
c^{2}(y) & =\rho^{(MC)}(y)-\rho^{(FC)}(y)\label{eq:c(y)_triplets}\\
e^{2}(y) & =1-\rho^{(MC)}(y)-\rho^{(FC)}(y)\label{eq:tm(y)_triplets}
\end{align}

\end{definition}


We next define a parametric class of multivariate densities for family data that can
easily be fit by maximum likelihood, allows for non-linear dependence,
and admits an analytical expression for the correlation curve~(\ref{correlationcurve}).

\section{Correlation and heritability curves for Gaussian mixtures} \label{section:Gauss}

Throughout this paper we denote by $\phi_d\!\left(\boldsymbol{y};\boldsymbol{\mu},\boldsymbol{\Sigma}\right)$ a $d$ dimensional Gaussian density,
evaluated at $\boldsymbol{y}=(y_1,\ldots,y_d)$, and with mean vector~$\boldsymbol{\mu}$
and covariance matrix~$\boldsymbol{\Sigma}$. We will only use $d=1,2,3$.

Consider the observed trait vector $\boldsymbol{y}=(y_1,y_2)$ for a pair of family members. We assume that it follows a $m$-component Gaussian mixture
with density 
\begin{equation}
\sum_{k=1}^{m}p_{k}\phi_{2}\!\left(\boldsymbol{y};\boldsymbol{\mu}_{k},\boldsymbol{\Sigma}_{k}\right),\label{eq:gaussianmixture}
\end{equation}
where $\sum_{k=1}^{m}p_{k}=1$. The mean and covariance structure 
of the the $k$th mixture component is taken to be
\begin{eqnarray}
\boldsymbol{\mu}_{k}=(\mu_{k},\mu_{k}),\quad\boldsymbol{\Sigma}_{k}=\begin{pmatrix}\sigma_{k}^{2} & \sigma_{k}^{2}\rho_{k}\\
\sigma_{k}^{2}\rho_{k} & \sigma_{k}^{2}
\end{pmatrix},\label{def:mu_Sigma}
\end{eqnarray}
where $\rho_{k}\in(-1,1)$ is the correlation parameter. The components
of the mixture are ordered such that $\sigma_{1}\leq\dots\leq\sigma_{m}$.
If $\sigma_{q}=\sigma_{q+1}=\dots=\sigma_{m}$ for some $q<m$, then
we order the components in ascending order with respect of the means,
i.e. $\mu_{q}<\dots<\mu_{m}$. Note that under the above constraints
on $\boldsymbol{\mu}_{k}$ and $\boldsymbol{\Sigma}_{k}$, 
the exchangeability condition \eqref{def:exchangability}
is satisfied. In addition, $Y_{1}$ and $Y_{2}$ have the same marginal
distribution, with marginal density

\begin{equation}
g(y)=\sum_{k=1}^{m}g_{k}(y)
\label{eq:marginal_mixture}\end{equation}
as the sum over the individual (weighted) components $g_{k}(y):=p_{k}\phi_{1}\!\left(y;\mu_{k},\sigma_{k}^{2}\right)$.
The (total) marginal mean, marginal variance, and correlation are given by
\begin{equation}
\mu=\sum_{k=1}^{m}p_{k}\mu_{k}, \quad
\sigma^{2}=\sum_{k=1}^{m}p_{k}\left[\sigma_{k}^{2}+(\mu_{k}-\mu)^{2}\right]
\quad\text{and}\quad
\rho=\sigma^{-2}\sum_{k=1}^{m}p_{k}\left[\rho_k\sigma_{k}^{2}+(\mu_{k}-\mu)^{2}\right].
\label{Globalsigma_mean}
\end{equation}

We next derive local versions of $\mu$ and $\sigma$. Let $\delta$
be a latent variable with $P(\delta=k)=p_{k}$, $k=1,\dots m$, showing
which mixture component is realized. From Bayes' rule, it follows
that the distribution of $\delta\mid Y_{2}=y$ is given as 
\begin{align}
p_{k}^{*}(y):=P(\delta=k\mid Y_{2}=y) & =\frac{g_{k}(y)}{g(y)}.\label{pikappa}
\end{align}
Also, by the assumed normality of each mixture component, it follows that 
\[
\mu_{k}(y):=\E(Y_{1}\mid Y_{2}=y,\delta=k)=\mu_{k}+\rho_{k}\cdot(y-\mu_{k}),
\]
i.e.~$\mu_{k}(y)$ is a line with slope $\rho_{k}$, going through
the point $(\mu_{k},\mu_{k})$. By the law of total expectation, 
\begin{equation}
\begin{split}\mu(y) & :=\E(Y_{1}|Y_{2}=y)=\E\left[\E\left(Y_{1}\mid Y_{2}=y,\delta\right)\mid Y_{2}=y\right]\\
 & =\sum_{k=1}^{m}p_{k}^{*}(y)\mu_{k}(y).
\end{split}
\label{expectation}
\end{equation}
Similarly, by the law of total variance: 
\begin{equation}
\begin{split}\sigma^{2}(y) & :=\Var(Y_{1}|Y_{2}=y)\\
 &=\E\left[\Var\left(Y_{1}\mid Y_{2}=y,\delta=k\right)\mid Y_{2}=y\right]\\
 &\phantom{=}+\Var\left[\E\left(Y_{1}\mid Y_{2}=y,\delta=k\right)\mid Y_{2}=y\right]\\
 & =\E\left(\sigma_{\delta}^{2}(1-\rho_{\delta}^{2})\mid Y_{2}=y\right)+\Var\left(\mu_{\delta}(y)\mid Y_{2}=y\right)\\
 & =\sum_{k=1}^{m}p_{k}^{*}(y)\left[\sigma_{k}^{2}(1-\rho_{k}^{2})+\left[\mu_{k}(y)-\mu(y)\right]^{2}\right].
\end{split}
\label{eq:sigma}
\end{equation}
We are now ready to give the expression for $\beta(y)=\mu'(y)$, to
be used in the correlation curve~\eqref{correlationcurve} for the
mixture distribution.\begin{proposition} \label{derivative} Define
\[
d_{k}(y):=-(y-\mu_{k})/\sigma_{k}^{2}.
\]
Then, 
\begin{align}
\beta(y)= & \sum_{k=1}^{m}p_{k}^{*}(y)\left[\rho_{k}+\left(\mu_{k}(y)-\mu(y)\right)d_{k}(y)\right],\label{eq:beta}
\end{align}
where $p_{k}^{*}(y)$ is given by (\ref{pikappa}).
\end{proposition}

\begin{proof}
See Appendix.
\end{proof}

Notice that when there is only a single mixture component $(m=1)$, yielding a bivariate Gaussian distribution, the above expressions reduce to $\sigma = \sigma_1$, $\mu(y)=\mu_1$, $\sigma^2(y) = \sigma^2_{1}(1 - \rho^2_{1})$ and $\beta(y)=\rho_1$. Inserting these expressions in (\ref{correlationcurve}) we get a constant correlation curve, $\rho(y)=\rho_1$ for every $y$. Hence, if $m=1$ the heritability curve $a^2(y)$, given by \eqref{eq:a(y)_twins} 
or \eqref{eq:a(y)_triplets}, reduces to the ordinary heritability coefficient $a^2$.

\subsection{Properties of the correlation curve under a Gaussian mixture} \label{sec:asym}

It is of interest to investigate the asymptotic behaviour of $\rho(y)$ as $y \rightarrow \pm \infty$ under the mixture (\ref{eq:gaussianmixture}) since this can be used to evaluate the asymptotic behaviour of the heritability curve $a^2(y)$, which in general will depend on the family design. We state the result in the following theorem, which also includes the limit behaviour of $\beta(y)$ and  $\sigma^2(y)$.

Intuitively, a one-dimensional mixture distribution is asymptotically
dominated in the tails by the component with the largest variance;
if two or more components all share the largest variance, the sizes
of the mean values come into play, with the component with the smallest
mean value dominating when $y\to-\infty$, and the largest when $y\to+\infty.$
While this in itself is fairly obvious, we here use it to develop
the resulting asymptotic behavior of $\beta(y)$, $\sigma^{2}(y)$,
and $\rho(y)$. 

We consider the following two cases: Recall the ordering $\sigma_{1}^{2}\leq\cdots\leq\sigma_{m}^{2}$,
and define $q=\min\left\{ l:\sigma_{l}^{2}=\sigma_{m}^{2}\right\} $.
We define Case I as $q=m$. For the alternative, Case II, where $q<m$,
our conventions is that the mean values are then ordered such that
$\mu_{q}<\mu_{m}$. To simplify the notation, define the constant
$K$ as follows:
\[
\begin{array}{ccc}
\text{Case I }(q=m), & y\to\pm\infty, & K:=m,\\
\text{Case II }(q<m), & y\to-\infty, & K:=q,\\
\text{Case II }(q<m), & y\to+\infty, & K:=m.
\end{array}
\]

\begin{theorem}\label{asymptotic}

The asymptotic behavior of $\beta(y)$, $\sigma^{2}(y)$, and $\rho(y)$,
given by (\ref{eq:beta}), (\ref{eq:sigma}), and (\ref{correlationcurve}),
are
\begin{align}
\lim_{y}\beta(y) & =\rho_{K},  \nonumber \\
\lim_{y}\sigma^{2}(y) & =\sigma_{K}^{2}(1-\rho_{K}^{2}), \nonumber \\
\lim_{y}\rho(y) & =\tilde{\rho}_{K}:=\frac{\sigma\rho_{K}}{\left[\sigma^{2}\rho_{K}^{2}+\sigma_{K}^{2}(1-\rho_{K}^{2})\right]^{1/2}}. \label{eq:asymrho}
\end{align}
The global variance $\sigma^{2}$ is defined as in (\ref{Globalsigma_mean}).
\end{theorem}

\begin{proof}
See Appendix.
\end{proof} 

Theorem \ref{asymptotic} shows that $\beta(y)$, $\sigma^2(y)$, and $\rho(y)$ all stabilize to finite limits as $y \rightarrow \pm \infty$, and their behaviour is determined by the variance and correlation of mixture component $K$, in addition to the global variance $\sigma^2$.
In Case I we have that the asymptotic correlation is the same in both tails, 
as exemplified in Figure~\ref{Fig:asymptotic}a) where $K=3$ and $\tilde\rho_3 \approx 0.5$.

\begin{figure}[H] 	
  \begin{subfigure}{7cm}
    \centering\includegraphics[width=7cm]{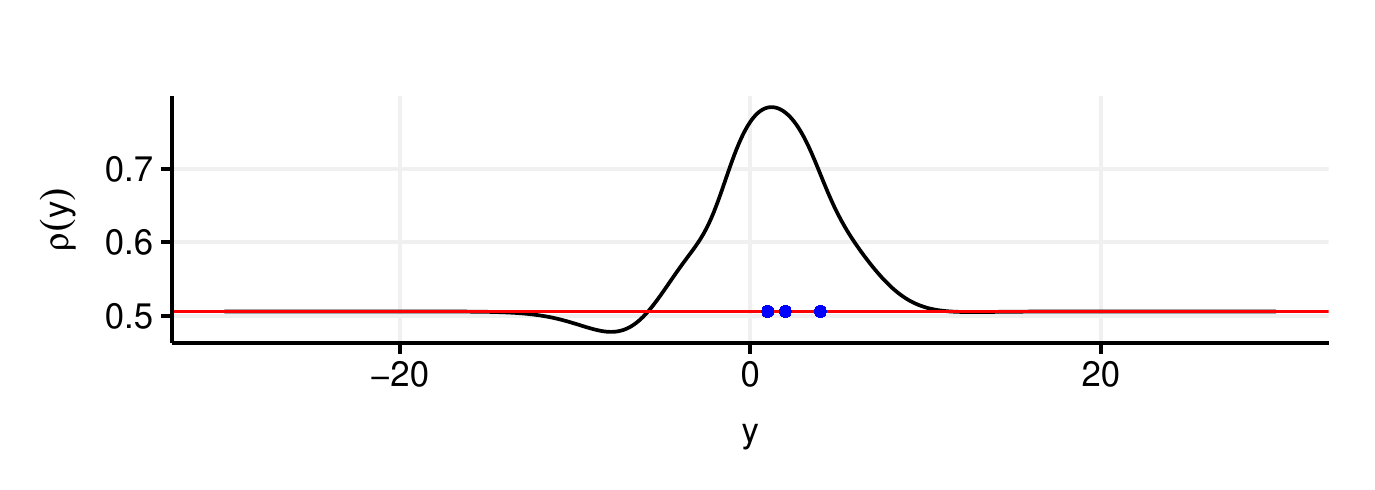}
  \end{subfigure}
  \begin{subfigure}{7cm}
    \centering\includegraphics[width=7cm]{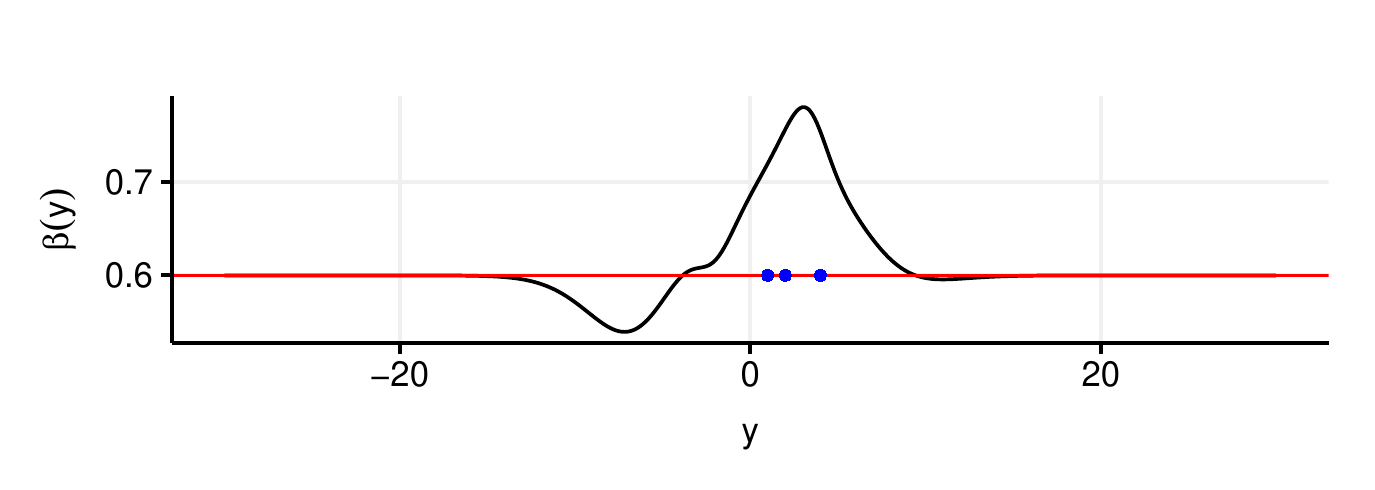}
  \end{subfigure}
 
  \begin{subfigure}{7cm}
    \centering\includegraphics[width=7cm]{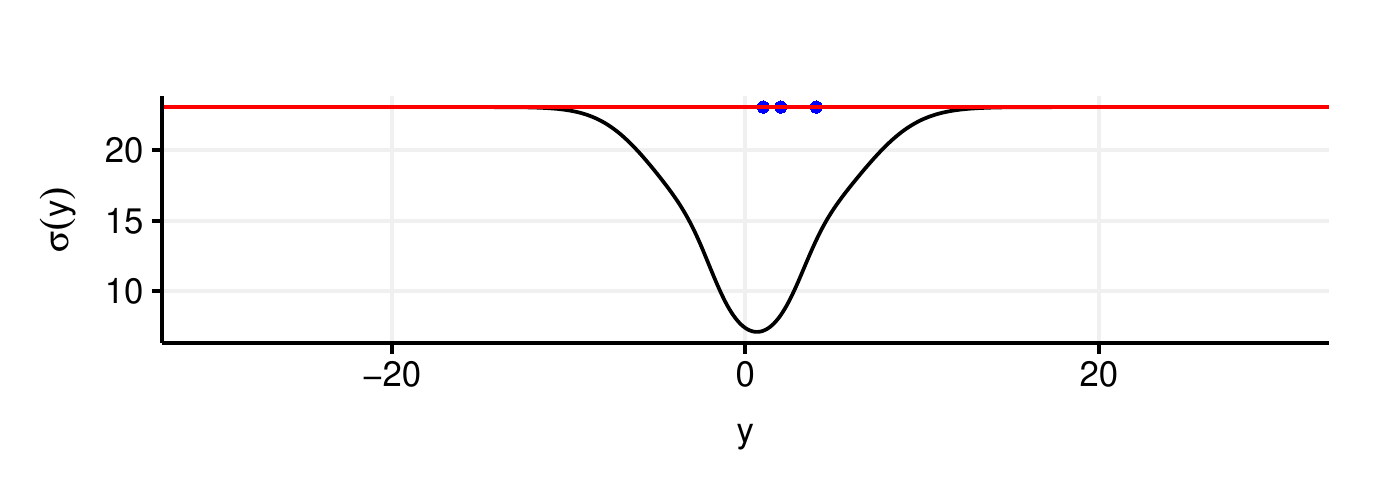}
  \end{subfigure}
  \begin{subfigure}{7cm}
  \centering\includegraphics[width=7cm]{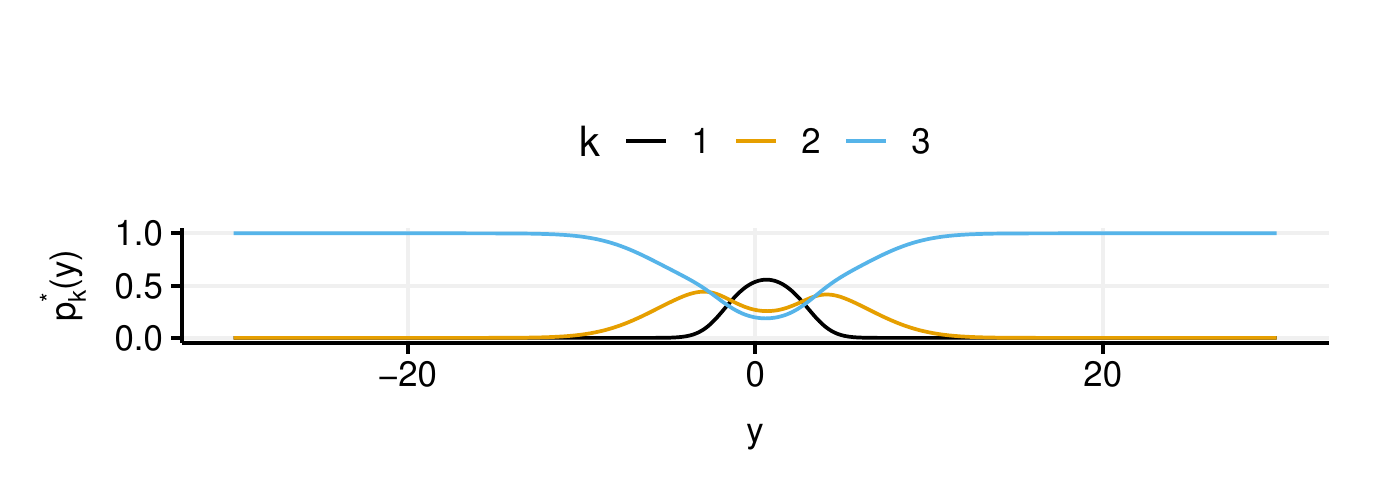}
  \end{subfigure}
  \caption{Illustration of the asymptotic tail behaviour (red line)
   of the correlation curve and its builing blocks under 
   a $m=3$ component mixture model:
  (a) $\rho(y)$, (b) $\beta(y)$, (c) $\sigma^2(y)$, and (d) $p_k^*(y)$. The mixture model has parameters $\left(\sigma_1, \sigma_2, \sigma_3 \right) = \left(2,4,6 \right)$, $\left( \mu_1, \mu_2, \mu_3 \right)= \left( 1,2,4 \right)$, $\left( \rho_1, \rho_2, \rho_3 \right) = \left(0.7, 0.8, 0.6 \right)$ and $\left(p_1, p_2, p_3 \right)= (0.3, 0.3, 0.4)$. 
  \label{Fig:asymptotic}}
\end{figure}

Identical correlations in both tails may seem unmotivated for family data.
Still, within the data range the correlation curve will be determined by all of the mixture components,
in accordance with \eqref{correlationcurve}, which allows for different behaviour in the tails.

Case II, on the other hand, allows for different asymptotic correlation in the left and right tail, with the differences being the use of $\rho_n$ versus $\rho_m$ in \eqref{eq:asymrho}.

Theorem~\ref{asymptotic} is further illustrated in Figure \ref{Fig:asymptotic} showing the limiting behaviour of 
$\beta(y)$, $\sigma^2(y)$, and $\rho(y)$ for a three-component mixture under Case~I. Note that
the limiting correlation satisfies $\tilde\rho_3<\min(\rho_1,\rho_2,\rho_3)$ for the parameter values used
in the figure. This is counter-intuitive because the posterior probability $p_3^{*}(y)$ approaches~1 
in the tails (upper left panel), but still the limiting correlation is not simply $\rho_3$. The peak in correlation
around $\mu_2=2$ is reasonable as the second component has the highest $\rho$.

\subsubsection{The case of equal $\sigma_k$'s}
It is worth studying the special case that
$\sigma_1=\sigma_2=\cdots=\sigma_m$, with their common value denoted by $\sigma_0$.
This is Case~II of Theorem~\ref{asymptotic} with $q=1$.
From~\eqref{Globalsigma_mean}
we get $\sigma^2=\sigma_0^2+\sigma_\mu^2$, where
\begin{equation}
\sigma_\mu^2 = \sum_{k=1}^{m}p_{k}(\mu_{k}-\mu)^{2},
\label{def:sigma2_mu}
\end{equation}
which is the variance due to differences in locations of mixture components.
Recall the convention that the mixture components are ordered such that 
$\mu_1<\mu_2<\cdots<\mu_m$. We are now ready to state the following corollary
to Theorem~\ref{asymptotic}.
\begin{corollary}\label{prop1}
When $\sigma_1=\cdots=\sigma_m$ the asymptotic behavior of 
$\rho(y)$, given by~\eqref{correlationcurve}, is	
\begin{equation}
\lim_{y\rightarrow-\infty}\rho(y)  =\rho_1\sqrt\frac{1+\gamma}{1+\gamma\rho_1^2}
\quad\hbox{and}\quad
\lim_{y\rightarrow\infty}\rho(y)  =\rho_m\sqrt\frac{1+\gamma}{1+\gamma\rho_m^2}, 
\end{equation}
where $\gamma=\sigma_\mu^2/\sigma_0^2$ is the ratio of between and 
within-component variance in the Gaussian mixture.
\end{corollary}
The limiting correlations always exceed (in absolute value) $\rho_1$ and $\rho_m$,
respectively. When $\gamma\rightarrow\infty$, i.e.~the mixture components gets
increasingly spread out, both limits approach 1 in absolute value.

\subsection{Estimation}\label{sec:est}

In this section we explain how to fit Gaussian mixtures to family data.
On one hand, they are fully parametric distributions, which can be exploited in estimation and inference. On the other hand, allowing the number of mixture components $m$ to grow, 
mixtures become increasingly flexible, which allows us to view them also
as nonparametric tools. In particular, Gaussian mixtures seem well
suited to model small perturbations from Gaussianity.

First, let $\boldsymbol{y}=(y_1, y_2, y_3)$ denote the trait vector for the mother-father-child trio,
which is assumed to have the following mixture density:
\begin{equation*}
\sum_{k=1}^m p_k \phi_3\!\left(\boldsymbol{y}; \boldsymbol{\mu}_{k}, \boldsymbol\Sigma_{k}\right).
\end{equation*}
Here  $\boldsymbol{\mu_k}$, $\boldsymbol{\Sigma_{k}}$ are structured in the following way:
\begin{eqnarray}
\boldsymbol{\mu}_k = (\mu_k, \mu_k, \mu_k),\quad \boldsymbol{\Sigma}_{k} =
 \begin{pmatrix}
  \sigma^2_{k} & \sigma^2_{k}\rho_{k}^{(MF)} & \sigma^2_{k}\rho_{k}^{(MC)}\\
  \sigma^2_{k}\rho_{k}^{(MF)} & \sigma^2_{k} & \sigma^2\rho_{k}^{(FC)}\\
  \sigma^2_{k}\rho_{k}^{(MC)} & \sigma^2_{k}\rho_{k}^{(FC)} & \sigma^2_{k}
 \end{pmatrix},
\label{eq:3mixture} 
\end{eqnarray}
where we use superscripts on the $\rho$'s to denote relationship.
Integrating the above joint density with respect to any one of the three family
members ($y_1$, $y_2$, or $y_3$) will result in the bivariate Gaussian mixture (\ref{eq:gaussianmixture}) from which we defined the correlation curve. The reason for performing joint estimation, rather than pairwise, is to optimally utilize the information
contained in mother-father-child trios. Note that the three marginals are identical
by construction, although the joint distribution is not exchangeable 
unless $\rho_{k}^{(MF)}=\rho_{k}^{(MC)}=\rho_{k}^{(FC)}$ for $k=1,\ldots,m$.

Given $n$ such trios, the parameters ($\mu_k$, $\sigma_k$, $\rho_k$, $p_k$) can be estimated by maximizing the following log-likelihood:
\begin{equation}
\log L = \sum_{i=1}^{n} \log \left[ \sum_{k=1}^m p_k \phi_3\!(\boldsymbol{y}_{i}; \boldsymbol{\mu}_k, \Sigma_{k}) \right].
\label{def:loglik_trio}
\end{equation}

\noindent Once the parameters are estimated, the heritability curve $a^2(y)$ can be obtained via the correlation curves 
as described in Definition~\ref{def:heritability_curve_trios}.

For twins, consider first a dizygotic pair with trait vector $\boldsymbol y = (y_{1}, y_{2})$.
The likelihood contribution from $n^{(MZ)}$ such pairs is:
\begin{equation}
\log L^{(MZ)} = \sum_{i=1}^{n^{(MZ)}}\log \sum_{k=1}^m p_k \phi_2\!(\boldsymbol{y}_{i}; \boldsymbol{\mu}_k, \boldsymbol{\Sigma}_{k}),
\label{def:loglik_twins}
\end{equation}
where $\boldsymbol{\mu}_k$ and $\boldsymbol{\Sigma}_{k}$ are structured as in~\eqref{def:mu_Sigma}.
The likelihood contribution of $n^{(DZ)}$ dizygotic twin pairs, $\log L^{(DZ)}$, is defined analogously using the same number $m$ of mixture components. The only parameters that differ between
 the MZ and DZ cases are the correlation parameters $\rho_k$ in~\eqref{def:mu_Sigma}. 
The fact that $p_k$, $\mu_k$, and $\sigma_k$ are shared across 
the MZ and DZ mixtures, calls for using a combined log-likelihood
 $\log L= \log L^{(MZ)}+\log L^{(DZ)}$.
Once the parameters are estimated, the heritability curve $a^2(y)$ can be obtained via the correlation curves 
as described in Definition~\ref{def:heritability_curve_twins}.

Both of the log-likelihoods~\eqref{def:loglik_trio} and~\eqref{def:loglik_twins} will be maximized using the R-package TMB~\citep{kristensen2016tmb}. 
In TMB the (negative) log-likelihood is implemented as a C++ function, which is compiled and linked into the R session, where the standard function minimizer
 {\tt nlminb} is employed. In addition, TMB calculates the gradient and Hessian (1st and 2nd order derivatives) of the log-likelihood by
Automatic Differentiation~\citep{kristensen2016tmb}. Such derivative information can substantially speed up the minimizer and make it more robust. 
Finally, TMB uses derivatives to calculate the approximate standard deviation of any interest quantity, as a function of the parameters, using
the delta method. This feature of TMB will be used to estimate pointwise confidence intervals of correlation and heritability curves. 

For the purpose of selecting the number of mixture components, $m$,
we calculate both of the criteria $\text{AIC}=-2\log(L)+2Q$
and $\text{BIC}=-2\log(L)+\log(n)Q$ for each candidate model,
where $Q$ is the number of parameters and $\log(L)$ is obtained either from~\eqref{def:loglik_trio} 
or ~\eqref{def:loglik_twins}. Contributing to $Q$ is 
the total number of $p_k$'s, $\mu_k$'s, $\sigma_k$'s, and $\rho_k$'s, but due to the constraint $\sum_{k=1}^{m}p_k=1$ there are only $m-1$ 
free $p_k$'s. Hence, for the trio likelihood~\eqref{def:loglik_trio} we 
have $Q=6m-1$, while for the twin likelihood~\eqref{def:loglik_twins}, with different $\rho_k$ for MZ and DZ twins, we have $Q=5m-1$.
It is clear that for $\log(n)>2$, BIC will be more conservative than AIC, in
the sense of favoring smaller values of $m$. As will be shown below, the correlation curve tends to be more unstable (fluctuating) for larger values 
of~$m$. For this reason we will use BIC as our model selection criterion,
but we will still report AIC as a comparison.  

\section{Applications} \label{section:Appl}

\subsection{BMI of twins}

We use the ``twinData''  dataset found in the R-package ``OpenMx'' \citep{neale2016openmx}. As our response, we take BMI measurements (around age 18) for $n^{(MZ)}=534$ monozygotic
and $n^{(DZ)}=328$ dizygotic female-female twin pairs. 
Table~\ref{Table: twin data} compares models in the range $1\leq m\leq 5$,
and it is seen that the pure bivariate Gaussian model ($m=1$) fits considerably worse than any of
the mixture models ($m>1$).
The lowest AIC and BIC values occur for $m=5$ and $m=2$, respectively,
but it is seen that AIC is almost indecisive between models with $m>1$.
Due to its heavier penalization, 
$\log\left(n^{(MZ)}+n^{(DZ)}\right)=\log(862)=6.8$, of the number of parameters,
BIC more clearly favours $m=2$.
According to our decision to base model selection on BIC, we
choose the model with $m=2$. 

\begin{table}[!htp]
\centering
\begin{tabular}{ccrr}
  \hline
$m$ & no. of parameters & AIC & BIC \\ 
  \hline
1 & 4 & 259.4 & 227.6 \\ 
  2 & 9 & 8.0 & 0 \\ 
  3 & 14 & 2.8 & 18.5 \\ 
  4 & 19 & 6.5 & 46.0 \\ 
  5 & 24 & 0 & 63.3 \\ 
   \hline
\end{tabular}
\caption{Model comparison for the twin BMI data, where $m$ is the number of mixture components and $5m-1$ is the number of parameters in the model.
AIC and BIC values are relative to the best fitting models (respectively, $m=5$ and $m=2$).
\label{Table: twin data}}
\end{table}

Table~\ref{Table:estimation2} shows the parameter estimates.
The first mixture component is dominating with $p_1=0.81$.
For MZ twins there is high correlation ($\rho_k$) within in each of the
two components, while for DZ twins  $\rho_2$ is close to zero.
The (global) correlations for the mixtures as a whole,
matches exactly the empirical Pearson correlations,
which are $0.78$ (MZ) and $0.30$ (DZ), respectively.

\begin{table}[!htp]
\centering
\begin{tabular}{crrr}
  \hline
Parameters & $k=1$ & $k=2$ & Global \\ 
  \hline
$\mu_k$ & 21.20 & 22.20 & 21.39 \\ 
$\sigma_k$ &  0.63 & 1.26 & 0.88 \\ 
$\rho^{(MZ)}_k$ &  0.75 & 0.70 & 0.78 \\ 
$\rho^{(DZ)}_k$ &  0.28 & $-$0.04 & 0.30 \\ 
$p_k$ &  0.81 & 0.19 &  \\ 
   \hline
\end{tabular}\caption{Parameter estimates for the chosen Gaussian mixture ($m=2$) for the twin data. The mixture components are ordered according to the value of $\sigma_k$. The global quantities, $\mu$, 
$\sigma$, $\rho^{(MZ)}$ and $\rho^{(DZ)}$ are calculated from~\eqref{Globalsigma_mean}. }
\label{Table:estimation2}
\end{table}

\begin{figure}[!htp] 
        \centering
\includegraphics{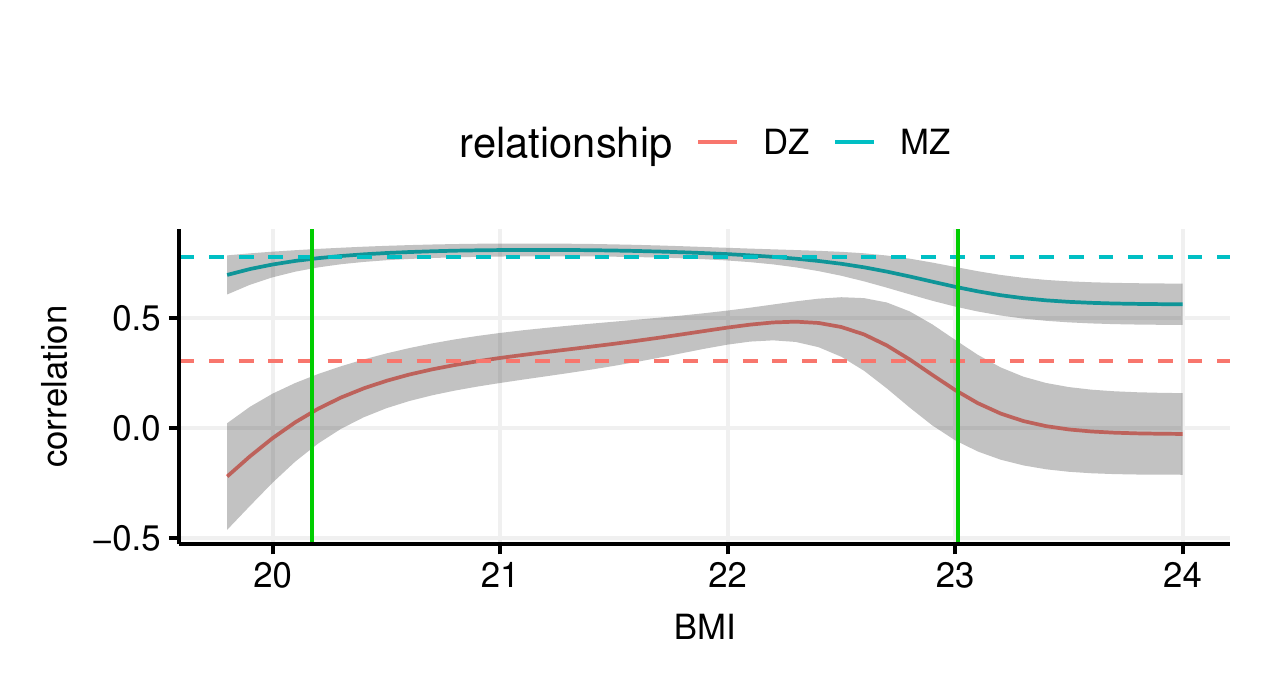}
       \caption{Estimated monozygotic (MZ) and dizygotic (DZ) twins correlation curves for the BMI data, with pointwise 95\% confidence
	   intervals (in grey). The dashed lines display the (overall) Pearson correlation within MZ and DZ twin pairs, respectively. The vertical green lines represent the $0.05$ and $0.95$ quantiles of the data.}
        \label{fig:correlation_twins}
\end{figure}

Figure~\ref{fig:correlation_twins} displays the estimated correlation curve for both MZ and DZ twins, using the parameter values from Table~\ref{Table:estimation2}. Also shown are 95\% confidence intervals calculated using the delta method. 
Both correlation curves are fairly flat within the center 90\% data range (represented by
the two vertical green bars), while they both drop for low and high BMI. 
This yields (Figure~\ref{fig:herandenv_twins}) an estimated heritability curve $a^2(y)$ that does not differ significantly (except maybe around $y=22.3$) from the classical heritability coefficient~\eqref{eq:ADE_moment}.

The TMB (R and C++) code used to produce the parameter estimates in Table~\ref{Table:estimation2} plots in Figure~\ref{fig:herandenv_twins} is available from \url{https://github.com/skaug/Supplementary}.

\begin{figure}[ht] 
        \centering
\includegraphics{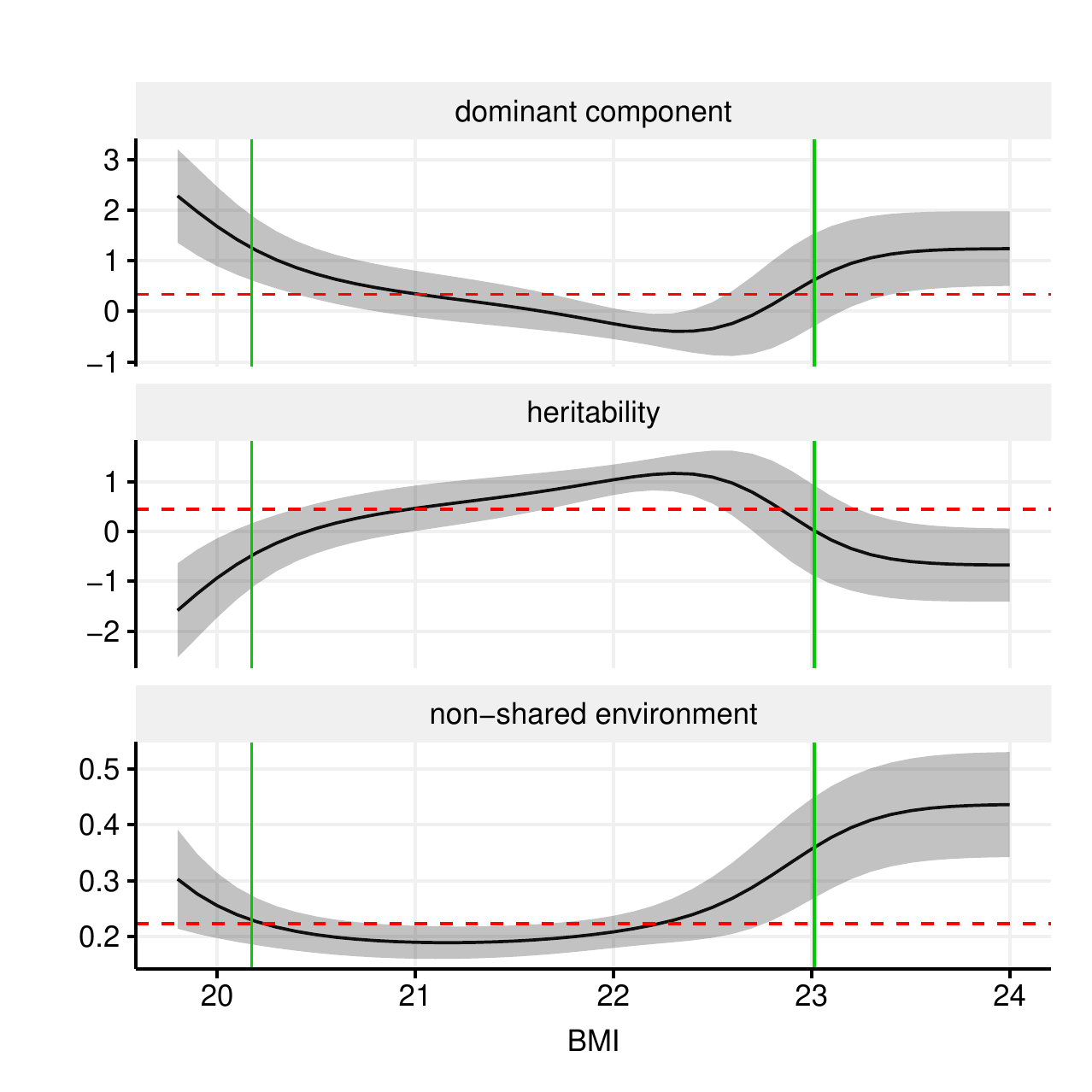}
       \caption{Estimated dominant genetic component $d^2(y)$, heritability curve $a^2(y)$,  and environment curve $c^2(y)$ for the BMI data under the ADE model (Definition~\ref{def:heritability_curve_twins}), with pointwise 95\% confidence
	   intervals (in grey). The red dashed lines display the classical estimates of dominant component, heritability, and environment, given by~\eqref{eq:ADE_moment}. The vertical green lines represent the $0.05$ and $0.95$ quantile in data.}
        \label{fig:herandenv_twins}
\end{figure}

\subsection{Birth weight of family trios}\label{subsec:BW}

To illustrate the family trio analyses, we used birth weights of $n=81,144$ complete mother--father--child trios. The data originally derived from the Medical Birth Registry of Norway, where the birth weight variables were added some random noise and rounded off to guarantee anonymity. The same data with some additional restrictions on parity, plurality, etc. were previously described and analyzed elsewhere~\citep{Magnus01}. The data were restricted to all births (mother, father, and child) taking place within the years 1967--1998.
Due to Norwegian ethical and legal restrictions, Norwegian data used in this study are available upon request to the Medical Birth Registry of Norway, the Norwegian Institute of Public Health. URL: https://www.fhi.no/hn/helseregistre-og-registre/mfr. Requests for data access can be directed to Datatilgang@fhi.no<mailto:Datatilgang@fhi.no>.

We did not have information about the gender of the child; hence, we performed a standardization of the data. We assumed a $50 \% $ sex ratio in the offspring, and introduced the quantity $D \triangleq \frac{1}{2} \left( \bar{y}_M - \bar{y}_F \right)$, where $\bar{y}_M$ is the mean of the birth weights of mothers, and $\bar{y}_F$ is the mean of the birth weights of fathers. We hence added $D$ to the father's weight and subtracted it to the mother's weight; in this way, the average among mothers and fathers is the same, and close (25g deviation) to the average in the offspring. This standardization is of little consequence to the end result.

Figure~\ref{fig:scatmat} summarizes the marginal and bivariate properties of the data. The marginal distributions are close to a Gaussian shape, but the left tail of the child birth weights is slightly heavier than the right tail. As suggested in the Introduction, this may be indicative of strong but rare factors dominating in producing the lowest birth weigths, which is what we will confirm in our analyses of local heritability below.

The scatter plots are roughly symmetric around the identity line, which is consistent with the exchangeability assumption made in Section~\ref{section:corr}. It should be noted, however, that the left hand tail of the marginal distributions is somewhat heavier in the children than in the parents; this is likely because parents are selected by the fact that they have children; it is known that individuals born with low birth weight have somewhat reduced fertility later in life. We have, however, not taken this into consideration in our model.

From the non-parametric regression (blue curve), it is clear that
there is no association between mother and father, which is reflected in the low Pearson correlation of $0.0209$.
For the two relationships involving the child, the non-parametric regression curve indicates a non-linear relationship,
particularly for mother-child. For birth weights less than 3000g there seems to be a low association, while for
larger birth weights the association is increasing.

 The Gaussian mixture (\ref{eq:gaussianmixture}) was fit by maximum likelihood for $m=1,\ldots,7$. We computed both AIC and BIC values for this model. According to the BIC criterion, the best fitting mixture has $m=4$ components (see Table~\ref{Table:trios_data}). Parameters estimates for this model are given in Table~\ref{Table:estimation4}. Figure \ref{fig:ellipses} shows the underlying mother-child pairs, overlaid by the five mixture components. 
 
 \begin{figure}[!htp] 
        \centering
\includegraphics{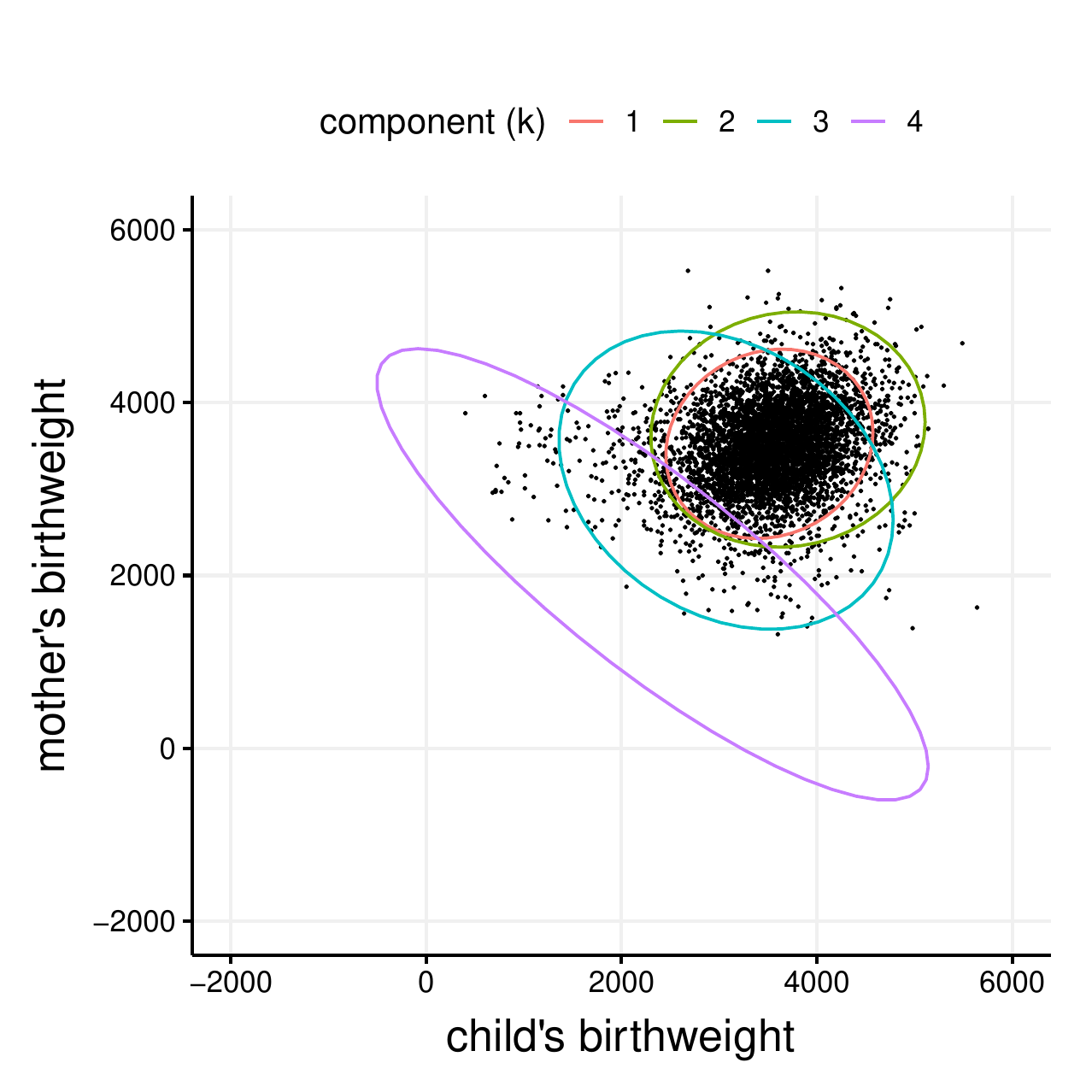}
       \caption{Birth weight (gram) of a random subset of $5000$ mother-child pairs taken
from Figure~\protect\ref{fig:scatmat}. Also shown are 95\% level curves (ellipses)
for each of the $m=4$ mixture components in Table~\protect\ref{Table:estimation4},
i.e.~each elipse include 95\% of the probability mass for that bivariate normal
component.
}
        \label{fig:ellipses}
\end{figure}

\begin{table}[ht]
\centering
\begin{tabular}{c c c c }
\hline
$m$ & no. parameters& $\Delta$ AIC  &$\Delta$ BIC  \\ \hline
$1$ & $5$  & 14848  & 14749 \\
  2 & 11 & 1148 & 904.4 \\ 
  3 & 17 & 480.4 & 292.5 \\ 
  4 & 23 & 132.1 & 0 \\ 
  5 & 29 & 109.7 & 33.5 \\ 
  6 & 35 & 36.3 & 16.0 \\ 
  7 & 41 & 0 & 35.5 \\ \hline
\end{tabular}
\caption{Model comparison for family trios, where $m$ is the number of mixture components. The total number of (free) parameters is $6m-1$,
counting all $p_k$, $\mu_k$, $\sigma_k$, $\rho_k^{(MC)}$, $\rho_k^{(FC)}$ and $\rho_k^{(MF)}$. 
AIC  and BIC values  are relative to the lowest one, represented in red. \label{Table:trios_data}}
\end{table}

The mother-child distribution is pear-shaped relative to a bivariate normal distribution,
with more spread around the identity line ($y_1=y_2$) for small
birth weights. The mixture model adapts to this shape by assigning
negative $\rho_k$'s to its two components ($k=3,4$) with the smallest $\mu_k$. The remaining two components ($k=1,2$),
which together constitute 87\% of the probability mass, form a bivariate 
distribution that is hard to distinguish visually from a Gaussian distribution.
The estimates of global correlation for the mixture in~Table~\ref{Table:trios_data},  closely match
the corresponding empirical Pearson correlations given in Figure~\ref{fig:scatmat} for MC, FC and MF pairs.
It is seen to fit the empirical marginals fairly well, and to posses a heavier left hand tail.  

\begin{table}[htp]
\centering
\begin{tabular}{crrrrr}
  \hline
Parameters & $k=1$ & $k=2$ & $k=3$ & $k=4$ & Global \\ 
  \hline
$\mu_k$ & 3516 & 3687& 3093 & 2243 & 3493 \\ 
$\sigma_k$ &  440.5 & 572.9 & 690.5 & 1116 & 555.0 \\ 
 $\rho^{(MC)}_k$ & 0.240 & 0.143 & $-$0.189 & $-$0.826 & 0.123 \\ 
 $\rho^{(FC)}_k$ & 0.134 & 0.053 & $-$0.254 & $-$0.845 & 0.201 \\ 
 $\rho^{(MF)}_k$ & $-$0.011 & $-$0.084 & $-$0.289 & 0.750 & 0.068	 \\ 
$p_k$ &  0.636 & 0.231 & 0.126 & 0.007 &  \\ 
   \hline
\end{tabular}\caption{Parameter estimates and standard deviations for the Gaussian mixture ($m=4$) fit to the mother–father–child trios. 
The mixture components are ordered according to the value of $\sigma_k$. The global quantities, $\mu$, 
$\sigma$, $\rho^{(MC)}$, $\rho^{(FC)}$ and $\rho^{(MF)}$ are calculated from~\eqref{Globalsigma_mean}. }
\label{Table:estimation4}
\end{table}

Figure \ref{fig:correlations} shows the two estimated correlation curves $\rho^{(FC)}(y)$ and $\rho^{(MC)}(y)$,
which are the components going into $a^2(y)$, $c^2(y)$, and $e^2(y)$, given respectively by \eqref{eq:a(y)_triplets}--\eqref{eq:tm(y)_triplets}.
Overall, the Pearson correlation and the correlation curves for MF exceed those for FC.
Both curves exceed their respective Pearson	 correlations in the center of the data, while they decrease 
for both low and high birth weights. The FC curve has its maximum somewhat to the left of the maximum of the MC curve. As a robustness check, we also computed the local Gaussian correlations ~\citep{Tjostheim2013} between mother and child as displayed in Figure \ref{fig:locgauss}. These exhibit the same behaviour as the correlation curve;  large values in the center of the data which are decreasing towards both tails.
Figure \ref{fig:herandenv} shows heritability and environment curves. 
The overall conclusion is that variation in birth weight is mostly attributable to environment,
which was also seen in previous publications~\citep{Magnus01,lunde_genetic_2007,Gjessing08}, and is reflected in the classical measures of heritability $a^2=0.246$ and environment $c^2=0.754$,
and the variation in the corresponding curves.

\begin{figure}[H] 
        \centering
\includegraphics{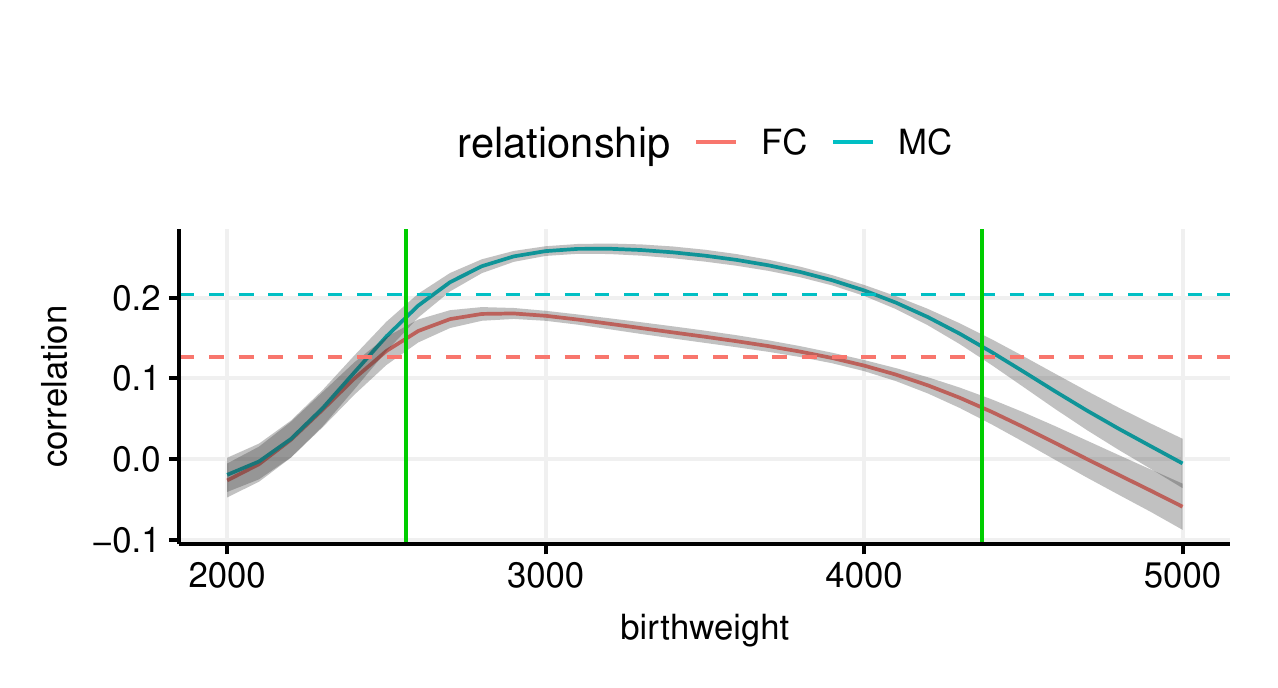}
       \caption{Estimated mother-child (MC) and father-child (FC) correlation curves for the Norwegian Birth Registry data, with pointwise 95\% confidence
	   intervals (in grey). The dashed lines display the (overall) Pearson correlation within MC and FC pairs, respectively. The vertical green lines represent the $0.05$ and $0.95$ quantiles of the data.}
        \label{fig:correlations}
\end{figure}

\begin{figure}[H]
       \centering
\includegraphics[scale=1]{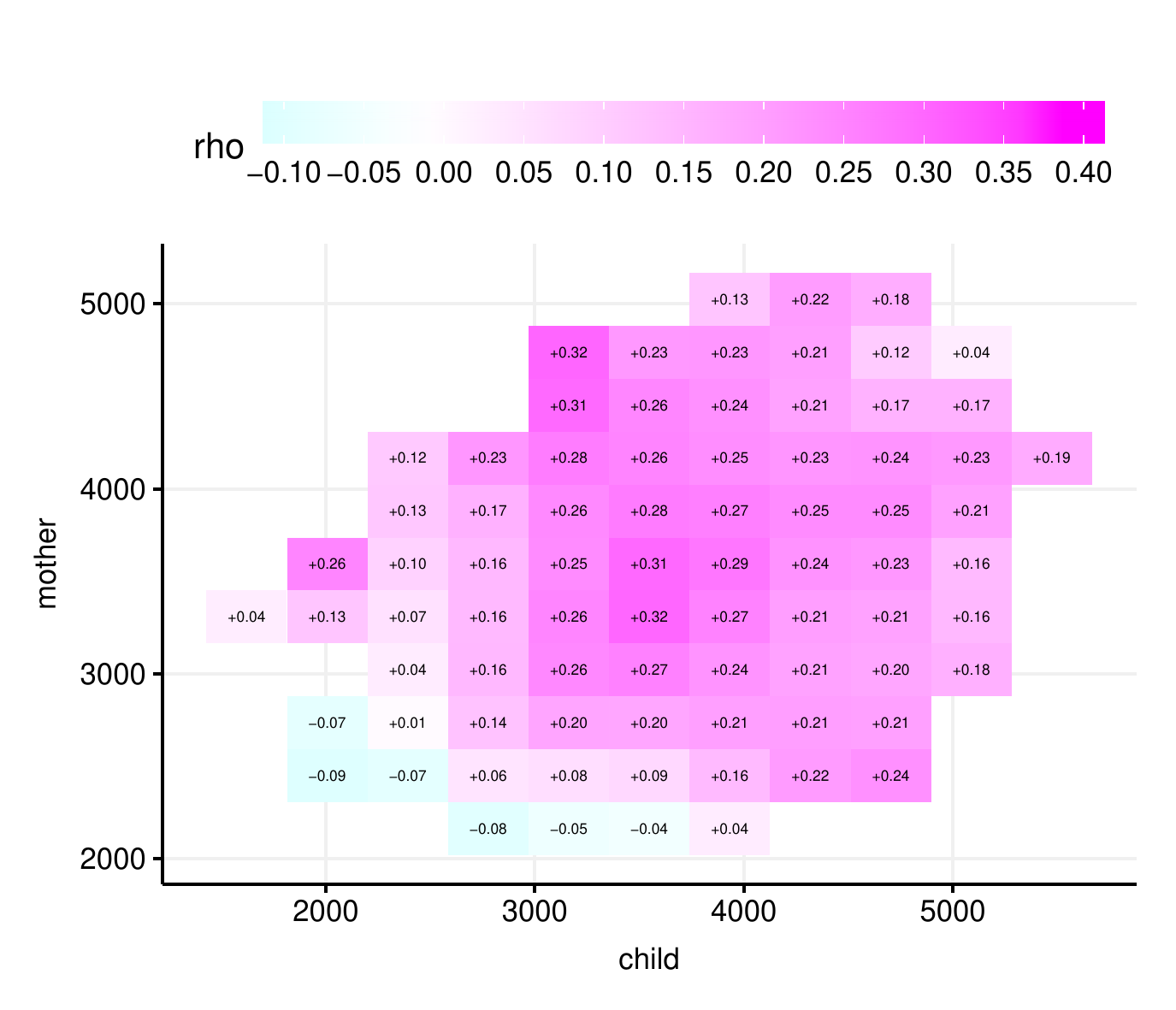}
      \caption{Estimated local Gaussian correlation between mother and child. Note that this correlation measure has two location arguments ($y_1$ and $y_2$). }
       \label{fig:locgauss}
\end{figure}

Recall that, under the assumed model \eqref{eq:ACE_MFC} the heritability curve $a^2(y)$ is completely determined by the FC correlation curve $\rho^{(FC)}(y)$. Since the FC correlation curve exceeds the Pearson FC correlation in the center of the data, the heritability curve also exceeds the classical heritability measure in the same region.

\begin{figure}[ht] 
        \centering
\includegraphics{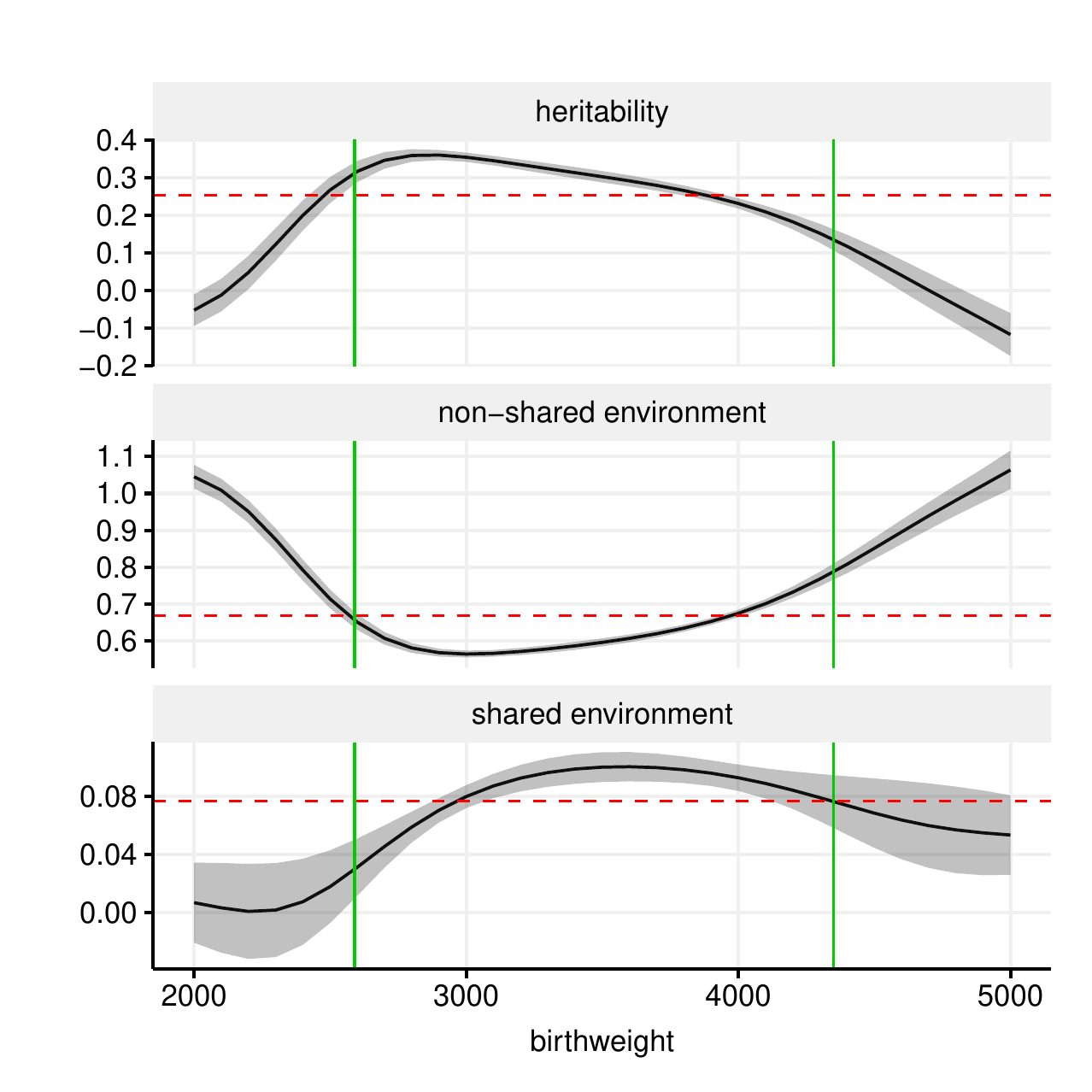}
       \caption{Estimated heritability curve $a^2(y)$, environment curve $c^2(y)$, and
       residual environment $e^2(y)$ for the Norwegian Birth Registry data under the ACE model (Definition~\ref{def:heritability_curve_trios}), with pointwise 95\% confidence
	   intervals (in grey). The red dashed lines display the classical
estimates of heritability and environment, i.e.~empirical 
versions of~\eqref{eq:ACE_MFC}. The vertical green lines represent the $0.05$ and $0.95$ quantiles of the data.}
        \label{fig:herandenv}
\end{figure}

\section{Discussion and conclusion}
We have provided closed-form expressions for the correlation curve
for exchangeable bivariate Gaussian mixtures. To our knowledge, this result is new and should be useful generally in situations where exchangeability 
can be assumed. Since differences in mean values may accounted for 
using a linear predictor like~\eqref{eq:multmean}, it is only 
exchangeability of the residuals, or the weaker condition~\eqref{def:exchangability}, that is required. In the context of our family data, the exchangeability assumption is rather
reasonable for twin data. In nuclear families, it is less obvious that parents and children have the exact same marginal distribution even when using covariates to adjust for systematic generational differences. With our generational birth weight data, we observe that the left hand tail in the parental distribution is smaller than among the children. As discussed in Subsection~\ref{subsec:BW}, this may well be a selection phenomenon; somebody born with a very low birth weight is less likely to become a parent, and are thus possibly under-represented in our data file. For instance, increased mortality among the smallest newborns is thought to lead to a selection pressure on the birth weight distribution over generations~\citep{in_cavalli-sforza_genetics_1999}.

A restriction of our model is that we have applied it only in situations with simple family structures where moment estimators of the heritability are explicit. In larger family structures, several pairwise relationships may provide information about the same heritability parameters. For instance, family trios with sibling data add the sibling correlation as a source of information~\citep{lunde_genetic_2007}. We will not discuss that issue further, but note that if pairwise correlation curves are estimated from larger data structures, weighted least squares estimation may provide a way of combining them into a common estimate of heritability curves~\citep{Gjessing08}.

In our twin BMI example, we chose the ADE model for the estimation since for the estimated overall correlations, $\rho^{(MZ)}>2\rho^{(DZ)}$. However, as seen in Figure~\ref{fig:herandenv_twins}, there are values for $y$ (the BMI) where the estimated $d^2(y)$ drops below zero. This indicates that in this region, the ACE model might be more appropriate. Note that there is no difficulty in letting the local heritability curves switch from an ADE model to an ACE model locally. In particular, we see that when $\rho^{(MZ)}=2\rho^{(DZ)}$, both (\ref{eq:Falconer}) and (\ref{eq:ADE_moment}) provide the same estimates for $a^2$ and $e^2$, and both $c^2$ and $d^2$ are estimated as zero. The estimated heritability curves would thus still be continuous if switching from one model to another.

The choice of Gaussian mixtures was made due to their flexibility, in the 
spirit of non-parametric estimation. Our approach is pragmatic in the sense that we have not attempted to interpret individual mixture components as sub-populations. One reason for this
is the negative estimates for some of the $\rho_k$ seen in both Table~\ref{Table:estimation2} 
and~\ref{Table:estimation4}, which would be hard to interpret biologically.  

On the other hand, Gaussian mixtures are fully parametric models, which
allows us to use the standard parametric toolbox. For instance, covariates
can easily enter the mean, as in~\eqref{eq:multmean}, and it would
also be straight forward to formulate model in which
the $\sigma_k$ were affect by family level covariates. 
A further benefit of having a parametric model is that we can select model complexity ($m$) based on standard AIC or BIC criteria. 

The parametric structure is also the basis for the results about the tail 
behaviour of the correlation curve in~Theorem~\ref{asymptotic}. While the center of 
the distribution may have sufficient data to allow stable non-parametric estimation of the heritability, the estimates in the tails are more dependent on the model structure. This is both a strength and a weakness of the mixture model. The heritability curves converge to constant values in the tails, which makes the estimates more stable; on the other hand, those estimates depend on the dominant mixture components in the tails, and the number and placement of mixture components may not always be clear cut.


There are also well known problems with Gaussian mixtures. Among these are local 
maxima on the likelihood surface ~\citep{baudry2015mixtures}, which can be explored by using different initial values for the numerical optimization.
We avoided the classical ``label switching'' problem by constraining the parameters
of the mixture ($\sigma$'s and $\mu$'s), but have nevertheless observed
some sensitivity of the parameter estimates in Table~\ref{Table:estimation4}. 
Although we cannot guarantee that we have found the global
optimum of the likelihood surface, the choice of model complexity ($m$) seems
to be robust to the choice of initial values. Similarly, the shape of the correlation curves (and consequently heritability and environment curves) are quite stable. 
A related problem is that of singlularity of the Fisher information matrix which
can occur for mixture models~\citep{drton2017bayesian}. This could potentially 
affect the validity of AIC and BIC criteria, as well as the standard deviations
based on the observed Fisher information that have been used throughout this paper.
Such standard deviations are produced automatically by TMB, and are very convenient in an exploratory phase, but we recommend that they
are validated by simulation (parametric bootstrap). 

\section{Acknowledgements}
This research was supported by Research Council of Norway grant 225912/F50 “Health Registries for Research” and the Centres of Excellence funding scheme (Grant 262700).

\clearpage
\bibliographystyle{apalike}
\bibliography{heritability_curves}

\begin{thebibliography}{}

\bibitem[Baudry and Celeux, 2015]{baudry2015mixtures}
Baudry, J.-P. and Celeux, G. (2015).
\newblock Em for mixtures.
\newblock {\em Statistics and computing}, 25(4):713--726.

\bibitem[Bender and Orszag, 2013]{bender_advanced_2013}
Bender, C.~M. and Orszag, S.~A. (2013).
\newblock {\em Advanced {Mathematical} {Methods} for {Scientists} and
  {Engineers} {I}: {Asymptotic} {Methods} and {Perturbation} {Theory}}.
\newblock Springer Science \& Business Media.
\newblock Google-Books-ID: xz0mBQAAQBAJ.

\bibitem[Bjerve and Doksum, 1993]{bjerve93}
Bjerve, S. and Doksum, K. (1993).
\newblock Correlation curves: Measures of association as functions of covariate
  values.
\newblock {\em Ann. Statist.}, 21(2):890--902.

\bibitem[Bulmer, 1985]{bulmer_mathematical_1985}
Bulmer, M.~G. (1985).
\newblock {\em The {Mathematical} {Theory} of {Quantitative} {Genetics}}.
\newblock Clarendon Press.

\bibitem[Cavalli-Sforza and Bodmer, 1999]{in_cavalli-sforza_genetics_1999}
Cavalli-Sforza, L.~L. and Bodmer, W.~F. (1999).
\newblock {\em The {Genetics} of {Human} {Populations}}, pages 612--614.
\newblock Courier Corporation.

\bibitem[Cherny et~al., 1992a]{cherny1992differential}
Cherny, S., Cardon, L., Fulker, D.~W., and DeFries, J. (1992a).
\newblock Differential heritability across levels of cognitive ability.
\newblock {\em Behavior genetics}, 22(2):153--162.

\bibitem[Cherny et~al., 1992b]{cherny_differential_1992}
Cherny, S.~S., Cardon, L.~R., Fulker, D.~W., and DeFries, J.~C. (1992b).
\newblock Differential heritability across levels of cognitive ability.
\newblock {\em Behavior Genetics}, 22(2):153--162.

\bibitem[DeFries and Fulker, 1985]{defries1985multiple}
DeFries, J.~C. and Fulker, D.~W. (1985).
\newblock Multiple regression analysis of twin data.
\newblock {\em Behavior genetics}, 15(5):467--473.

\bibitem[DeFries and Fulker, 1988]{defries1988multiple}
DeFries, J.~C. and Fulker, D.~W. (1988).
\newblock Multiple regression analysis of twin data: Etiology of deviant scores
  versus individual differences.
\newblock {\em Acta Geneticae Medicae et Gemellologiae: Twin Research},
  37(3-4):205--216.

\bibitem[Doksum et~al., 1994]{doksum_correlation_1994}
Doksum, K., Blyth, S., Bradlow, E., Meng, X.-L., and Zhao, H. (1994).
\newblock Correlation {Curves} as {Local} {Measures} of {Variance} {Explained}
  by {Regression}.
\newblock {\em Journal of the American Statistical Association},
  89(426):571--582.

\bibitem[Drton and Plummer, 2017]{drton2017bayesian}
Drton, M. and Plummer, M. (2017).
\newblock A bayesian information criterion for singular models.
\newblock {\em Journal of the Royal Statistical Society: Series B (Statistical
  Methodology)}, 79(2):323--380.

\bibitem[Falconer, 1960]{falconer60}
Falconer, D.~S. (1960).
\newblock {\em Introduction to quantitative genetics}.
\newblock Oliver And Boyd; Edinburgh; London.

\bibitem[Fisher, 1919]{fisher_xvcorrelation_1919}
Fisher, R.~A. (1919).
\newblock {XV}.—{The} {Correlation} between {Relatives} on the {Supposition}
  of {Mendelian} {Inheritance}.
\newblock {\em Earth and Environmental Science Transactions of The Royal
  Society of Edinburgh}, 52(2):399--433.
\newblock Publisher: Royal Society of Edinburgh Scotland Foundation.

\bibitem[Gjessing and Lie, 2008]{Gjessing08}
Gjessing, H.~K. and Lie, R.~T. (2008).
\newblock Biometrical modelling in genetics: are complex traits too complex?
\newblock {\em Statistical Methods in Medical Research}, 17(1):75--96.
\newblock PMID: 17855744.

\bibitem[Holland and Wang, 1987]{Holland1987}
Holland, P.~W. and Wang, Y.~J. (1987).
\newblock Dependence function for continuous bivariate densities.
\newblock {\em Communications in Statistics - Theory and Methods},
  16(3):863--876.

\bibitem[Hopper, 2002]{Hopper2002}
Hopper, J.~L. (2002).
\newblock {Heritability}.
\newblock In Elston, R., Olson, J., and Palmer, L., editors, {\em
  Biostatistical Genetics and Genetic Epidemiology}, Wiley reference series in
  biostatistics, pages {371--372}. Wiley, West Sussex, UK.

\bibitem[Hopper and Visscher, 2002]{HopperVisscher2002}
Hopper, J.~L. and Visscher, P.~M. (2002).
\newblock {Genetic Correlations and Covariances}.
\newblock In Elston, R., Olson, J., and Palmer, L., editors, {\em
  Biostatistical Genetics and Genetic Epidemiology}, Wiley reference series in
  biostatistics, pages {327--331}. Wiley, West Sussex, UK.

\bibitem[Khoury et~al., 1993]{khoury_fundamentals_1993}
Khoury, M.~J., Beaty, T.~H., and Cohen, B.~H. (1993).
\newblock {\em Fundamentals of {Genetic} {Epidemiology}}.
\newblock Oxford University Press.

\bibitem[Kristensen et~al., 2016]{kristensen2016tmb}
Kristensen, K., Nielsen, A., Berg, C.~W., Skaug, H., and Bell, B.~M. (2016).
\newblock Tmb: Automatic differentiation and laplace approximation.
\newblock {\em Journal of Statistical Software}, 70(1):1--21.

\bibitem[LaBuda et~al., 1986]{labuda1986multiple}
LaBuda, M.~C., DeFries, J., Fulker, D.~W., and Rao, D. (1986).
\newblock Multiple regression analysis of twin data obtained from selected
  samples.
\newblock {\em Genetic epidemiology}, 3(6):425--433.

\bibitem[Logan et~al., 2012a]{logan2012heritability}
Logan, J.~A., Petrill, S.~A., Hart, S.~A., Schatschneider, C., Thompson, L.~A.,
  Deater-Deckard, K., DeThorne, L.~S., and Bartlett, C. (2012a).
\newblock Heritability across the distribution: An application of quantile
  regression.
\newblock {\em Behavior genetics}, 42(2):256--267.

\bibitem[Logan et~al., 2012b]{logan_heritability_2012}
Logan, J.~A., Petrill, S.~A., Hart, S.~A., Schatschneider, C., Thompson, L.~A.,
  Deater-Deckard, K., DeThorne, L.~S., and Bartlett, C. (2012b).
\newblock Heritability {Across} the {Distribution}: {An} {Application} of
  {Quantile} {Regression}.
\newblock {\em Behavior genetics}, 42(2):256--267.

\bibitem[Lunde et~al., 2007]{lunde_genetic_2007}
Lunde, A., Melve, K.~K., Gjessing, H.~K., Skj{\ae}rven, R., and Irgens, L.~M.
  (2007).
\newblock Genetic and {Environmental} {Influences} on {Birth} {Weight}, {Birth}
  {Length}, {Head} {Circumference}, and {Gestational} {Age} by {Use} of
  {Population}-based {Parent}-{Offspring} {Data}.
\newblock {\em American Journal of Epidemiology}, 165(7):734--741.

\bibitem[Magnus et~al., 2001]{Magnus01}
Magnus, P., Gjessing, H.~K., Skrondal, A., and Skj{\ae}rven, R. (2001).
\newblock Paternal contribution to birth weight.
\newblock {\em Journal of Epidemiology \& Community Health}, 55(12):873--877.

\bibitem[McCulloch and Neuhaus, 2001]{mcculloch01}
McCulloch, C.~E. and Neuhaus, J.~M. (2001).
\newblock {\em Generalized linear mixed models}.
\newblock Wiley Online Library.

\bibitem[McLachlan and Peel, 2000]{mclachlan2000fmm}
McLachlan, G. and Peel, D. (2000).
\newblock {\em {Finite mixture models}}.
\newblock Wiley New York.

\bibitem[Neale, 2002]{Neale2002}
Neale, M.~C. (2002).
\newblock {Twin Analysis}.
\newblock In Elston, R., Olson, J., and Palmer, L., editors, {\em
  Biostatistical Genetics and Genetic Epidemiology}, Wiley reference series in
  biostatistics, pages {206--217}. Wiley, West Sussex, UK.

\bibitem[Neale et~al., 2016]{neale2016openmx}
Neale, M.~C., Hunter, M.~D., Pritikin, J.~N., Zahery, M., Brick, T.~R.,
  Kirkpatrick, R.~M., Estabrook, R., Bates, T.~C., Maes, H.~H., and Boker,
  S.~M. (2016).
\newblock Openmx 2.0: Extended structural equation and statistical modeling.
\newblock {\em Psychometrika}, 81(2):535--549.

\bibitem[Pawitan et~al., 2004]{pawitan2004}
Pawitan, Y., Reilly, M., Nilsson, E., Cnattingius, S., and Lichtenstein, P.
  (2004).
\newblock Estimation of genetic and environmental factors for binary traits
  using family data.
\newblock {\em Statistics in Medicine}, 23:449--465.

\bibitem[Rabe-Hesketh et~al., 2008]{rabe-hesketh_biometrical_2008}
Rabe-Hesketh, S., Skrondal, A., and Gjessing, H. (2008).
\newblock Biometrical modeling of twin and family data using standard mixed
  model software.
\newblock {\em Biometrics}, 64(1):280--288.

\bibitem[Tj{\o}stheim and Hufthammer, 2013]{Tjostheim2013}
Tj{\o}stheim, D. and Hufthammer, K.~O. (2013).
\newblock Local gaussian correlation: A new measure of dependence.
\newblock {\em Journal of Econometrics}, 172(1):33 -- 48.

\bibitem[Wright, 1920]{wright_relative_1920}
Wright, S. (1920).
\newblock The relative importance of heredity and environment in determining
  the piebald pattern of guinea-pigs.
\newblock {\em Proceedings of the National Academy of Sciences of the United
  States of America}, 6(6):320--332.
\newblock Publisher: National Academy of Sciences.

\bibitem[Wright, 1921]{wright21}
Wright, S. (1921).
\newblock Correlation and causation.
\newblock {\em Journal of agricultural research}, 20(7):557--585.

\end{thebibliography}

\clearpage

\appendix
\section{Proofs \label{app1}}

\begin{proof}
[Proof of Proposition \ref{derivative}]
Let $g(y)$, $g_{k}(y)$, $p_{k}^{*}(y)$ etc.~be defined as in Section~\ref{section:Gauss}. First,
note that
\[
\frac{g_{k}'(y)}{g_{k}(y)}=d_{k}(y).
\]
Furthermore, define
\[
d(y):=\sum_{i=1}^{m}p_{i}^{*}(y)d_{i}(y),
\]
i.e.~the weighted average of the $d_{i}(y)$'s. Then
\[
\frac{g'(y)}{g(y)}=\frac{\sum_{i=1}^{m}d_{i}(y)g_{i}(y)}{g(y)}=d(y).
\]
For any fraction $s(y)=a(y)/b(y)$ of differentiable functions, note
that the chain rule can be written as $\frac{s'(y)}{s(y)}=\frac{a'(y)}{a(y)}-\frac{b'(y)}{b(y)}$.
Thus,

\[
\frac{p_{k}^{*'}(y)}{p_{k}^{*}(y)}=\frac{g_{k}'(y)}{g_{k}(y)}-\frac{g'(y)}{g(y)}=d_{k}(y)-d(y).
\]
Recall from (\ref{expectation}) that $\mu(y)=\E\left[Y_{1}\mid Y_{2}=y\right]=\sum_{i=1}^{m}p_{i}^{*}(y)\mu_{i}(y)$
is the conditional expectation, 
\begin{align*}
\beta(y)=\mu'(y) & =\sum_{i=1}^{m}\left(p_{i}^{*}(y)\mu{}_{i}'(y)+p_{i}^{*'}(y)\mu_{i}(y)\right)\\
 & =\sum_{i=1}^{m}p_{i}^{*}(y)\left(\rho_{i}+\mu_{i}(y)\left(d_{i}(y)-d(y)\right)\right)\\
 & =\sum_{i=1}^{m}p_{i}^{*}(y)\left(\rho_{i}+\left(\mu_{i}(y)-\mu(y)\right)\left(d_{i}(y)-d(y)\right)\right)\\
 & =\sum_{i=1}^{m}p_{i}^{*}(y)\left(\rho_{i}+\left(\mu_{i}(y)-\mu(y)\right)d_{i}(y)\right),
\end{align*}
where we make use of $\sum_{i=1}^{m}p_{i}^{*}(y)\left(d_{i}(y)-d(y)\right)=0$ and  $\sum_{i=1}^{m}p_{i}^{*}(y)\left(\mu_{i}(y)-\mu(y)\right)=0$.

\end{proof}

\subsubsection{Proof of Theorem \ref{asymptotic} - asymptotic behavior of $\beta(y)$, $\sigma^2(y)$, and $\rho(y)$} \label{app:beta}

For two functions $a(y)$ and $b(y)$, as $y\to\infty$ (or $-\infty$),
we use the standard notation that $a(y)\sim b(y)$ means $\lim_{y\to\infty}a(y)/b(y)=1$,
and $a(y)\ll b(y)$ means $\lim_{y\to\infty}a(y)/b(y)=0$. Our proofs below follow mostly from standard theory on asymptotic behavior of real functions\cite{bender_advanced_2013}.

\paragraph{Asymptotic behavior of mixture components}

For one mixture component $g_{k}(y)$, the asymptotic behavior when
$y\to\pm\infty$ is 
\[
g_{k}(y)\sim C_{k}\exp\left(\frac{\mu_{k}}{\sigma_{k}}y-\frac{1}{2\sigma_{k}^{2}}y^{2}\right),
\]
for a constant $C_{k}$. Comparing two components $g_{k}(y)$ and
$g_{l}(y)$ with $\sigma_{k}^{2}<\sigma_{l}^{2}$, we clearly have
\begin{equation}
g_{k}(y)\ll g_{l}(y)\quad\text{as}\quad y\to\pm\infty\label{eq:pm}
\end{equation}
since the $y^{2}$-term dominates the asymptotics. If $\sigma_{k}^{2}=\sigma_{l}^{2}$,
assume that $\mu_{k}<\mu_{l}$. Then
\begin{equation}
g_{k}(y)\ll g_{l}(y)\quad\text{as}\quad y\to+\infty,\label{eq:p}
\end{equation}
and
\begin{equation}
g_{l}(y)\ll g_{k}(y)\quad\text{as}\quad y\to-\infty.\label{eq:m}
\end{equation}
Let $a_{k}(y)$ be non-zero polynomial functions in $y$ for $k=1,\ldots,m$.
Since polynomials are asymptotically dominated by exponentials of
polynomials, the products $g_{k}(y)a_{k}(y)$ are asymptotically ordered
in the same way as in (\ref{eq:pm}), (\ref{eq:p}), and (\ref{eq:m})
above.

\paragraph{Asymptotic behavior of mixtures}

Recall the definition of $K$ in Theorem~\ref{asymptotic}. The results
above apply directly to the sum $\sum_{k=1}^{m}g_{k}(y)a_{k}(y)$,
which will asymptotically follow the dominant term with $k=K$. I.e.,
\[
\sum_{k=1}^{m}g_{k}(y)a_{k}(y)\sim g_{K}(y)a_{K}(y).
\]
In particular, for the full density we get
\[
g(y)=\sum_{i=1}^{m}g_{i}(y)\sim g_{K}(y).
\]
Similarly, if $k\neq K$,
\begin{equation}
p_{k}^{*}(y)a_{k}(y)=\frac{g_{k}(y)a_{k}(y)}{g(y)}\to0,\label{eq:k.neq.K}
\end{equation}
and
\[
p_{K}^{*}(y)a_{K}(y)\sim a_{K}(y).
\]

\paragraph{Conditional mean $\mu(y)$}

Applying the above results to $\mu$, we obtain
\[
\mu(y)=\sum_{k=1}^{m}p_{k}^{*}(y)\mu_{k}(y)\sim\mu_{K}(y)\sim\rho_{K}\cdot y.
\]
Furthermore, letting $a_{k}(y)=\rho_{k}+\left(\mu_{k}(y)-\mu(y)\right)d_{k}(y)$,
we get
\[
\beta(y)=\sum_{k=1}^{m}p_{k}^{*}(y)a_{k}(y)\sim a_{K}(y).
\]
However, by~\ref{eq:k.neq.K}, 
\[
\left(\mu_{K}(y)-\mu(y)\right)d_{K}(y) =\sum_{k=1}^{m}p_{k}^{*}(y)(\mu_{K}(y)-\mu_{k}(y))d_{K}(y)\to 0
\]
since the $K$'th term vanishes. It follows that
\[
\beta(y)\sim a_{K}(y)\to\rho_{K}.
\]

\paragraph{Conditional variance $\sigma^{2}(y)$}

For the conditional variance, 
\begin{align*}
\sigma^{2}(y) & =\sum_{k=1}^{m}p_{k}^{*}(y)\left[\sigma_{k}^{2}(1-\rho_{k}^{2})+\left[\mu_{k}(y)-\mu(y)\right]^{2}\right]\\
 & \sim\sigma_{K}^{2}(1-\rho_{K}^{2}).
\end{align*}

\paragraph{Correlation curve $\rho(y)$}

Finally, the result for the correlation curve $\rho(y)$ follows directly
from the results for $\sigma^{2}(y)$ and $\beta(y)$.

\end{document}